\documentclass[final]{siamltex}
\usepackage{amsmath,amsfonts}
\usepackage{amssymb}
\usepackage{graphicx,tabularx,array,geometry}
\usepackage[ruled,vlined]{algorithm2e}
\usepackage{etoolbox}
\usepackage{xcolor}
\usepackage{bbm}

\newcommand{\vect}[1]{\boldsymbol{#1}}

\title{Precomputing strategy for Hamiltonian Monte Carlo Method based on regularity in parameter space}

\author{Cheng Zhang\footnotemark[1]\and Babak Shahbaba\footnotemark[2]\and Hongkai Zhao\footnotemark[1]}

\begin{document}


\maketitle
\footnotetext[1]{Department of Mathematics, University of California, Irvine, USA}
\footnotetext[2]{Department of Statistics, University of California, Irvine, USA}

\begin{abstract}
Markov Chain Monte Carlo (MCMC) algorithms play an important role in statistical inference problems dealing with intractable probability distributions. Recently, many MCMC algorithms such as Hamiltonian Monte Carlo (HMC) and Riemannian Manifold HMC have been proposed to provide distant proposals with high acceptance rate. These algorithms, however, tend to be computationally intensive which could limit their usefulness, especially for big data problems due to repetitive evaluations of functions and statistical quantities that depend on the data. This issue occurs in many statistic computing problems. In this paper, we propose a novel strategy that exploits smoothness (regularity) in parameter space to improve computational efficiency of MCMC algorithms. When evaluation of functions or statistical quantities are needed at a point in parameter space, interpolation from precomputed values or previous computed values is used. More specifically, we focus on Hamiltonian Monte Carlo (HMC) algorithms that use geometric information for faster exploration of probability distributions. Our proposed method is based on precomputing the required geometric information on a set of grids before running sampling algorithm and approximating the geometric information for the current location of the sampler using the precomputed information at nearby grids at each iteration of HMC. Sparse grid interpolation method is used for high dimensional problems. Tests on computational examples are shown to illustrate the advantages of our method. 
\end{abstract}



\pagestyle{myheadings}
\thispagestyle{plain}
\markboth{Cheng et al.}{Precomputing strategy for HMC}

\section{Introduction}
Many statistical and machine learning methods rely on costly iterative algorithms for optimization or sampling. One of the most computationally intensive components of these methods is repetitive evaluations of functions, their derivatives, geometric and statistical quantities that depend on the data. This is especially challenging in Big Data problems. To reduce the computation cost, one common approach is subsampling, which restricts the computation to a subset of the data or sample, such as stochastic gradient methods (See, for example, \cite{hoffmann10, wellingTeh11}). Another approach could be to find some computationally cheaper surrogate functions to substitute the expensive objective functions. (See, for example, \cite{rasmussen03, shahbabaSplitHMC}.) In this paper, we propose a different approach that explores smoothness or regularity in parameter space, which is true for most statistical models. When evaluation of functions or statistics quantities are needed at a point in parameter space, interpolation from precomputed values or previous computed values is used. Here, we mainly focus on a state-of-the-art class of Markov Chain and Monte Carlo (MCMC) sampling algorithms called Hamiltonian Monte Carlo (HMC). However, our proposed method could be extended to other computationally intensive, iterative algorithms commonly used in statistics and machine learning. 

MCMC was first introduced by Metropolis \cite{metropolis59} to simulate the distribution of states for a system of idealized molecules. Almost contemporarily, Alder and Wainwright \cite{alder59} proposed a deterministic approach to molecules simulation called molecular dynamics (MD). In the following decades, the MCMC and molecular dynamics approaches have continued to develop in their respective areas. In 1987, Duane, Kennedy, Pendleton and Roweth \cite{duane87} made a remarkable breakthrough by developing a Hybrid Monte Carlo (HMC) algorithm based on combining MCMC and molecular dynamics approaches. This is also known as Hamiltonian Monte Carlo (HMC) in the literature. Neal \cite{neal11} provided an extensive review of this method and presented several extensions. The basic idea is that starting from the current state of MCMC, one can use MD to generate trial moves (i.e., proposals within the Metropolis algorithm) that can move far from the current state (resulting in low autocorrelations) while keeping the acceptance probability high (by moving towards high probability regions). Therefore, the HMC sampling method can provide more rapid and efficient exploration of the parameter space than standard random walk proposals. However, HMC requires expensive gradient computations in order to simulate the Hamiltonian dynamics system. This could be infeasible in Big Data problems. Therefore, in recent years, there have been many attempts to improve computational efficiency of HMC and its variants. (See for example, \cite{wellingTeh11, ahmadian11, girolami11, hoffman11, shahbabaSplitHMC, beskos11, calderhead12, lan14, ahn13distributed, ahnShahbabaWelling14}.) One possible strategy is to use small subsets of data at each iteration \cite{wellingTeh11, shahbabaSplitHMC}. As an alternative approach, the precomputing strategies we propose here can reduce the computation cost of HMC while maintaining the overall efficiency of the method by exploiting smooth dependence of parameters that exists in typical statistical models. 

Before we present our method, we first briefly review HMC in the following section. We then present our proposed method in Section \ref{precompute} and evaluate its performance in Section \ref{experiments} using several examples. In Section \ref{approx}, we discuss an extension of our method that is faster computationally, but converges to an approximation of the target distribution. Finally, Section \ref{discussion} is devoted to discussion of future research directions and applications of our method in other algorithms. 

\section{Hamiltonian Monte Carlo}
In Bayesian Statistics, we are interested in sampling from the posterior distribution of the model parameters $q$ given the observed data, $Y=(y_1,y_2,\ldots,y_N)^T$,
\begin{equation}
P(q|Y) \propto \exp(-U(q)),
\end{equation}
where the potential energy function $U$ is defined as
\begin{equation}
U(q) = -\sum_{i=1}^N\log P(y_i|q) -\log P(q).
\end{equation}
Here, the first term is the negative log-likelihood, and $P(q)$ is the assumed prior on model parameters. The posterior distribution is almost always intractable. Therefore, Markov Chain Monte Carlo (MCMC) algorithms are typically used for sampling from the posterior distribution to perform statistical inference. We could for example use the Metropolis algorithm as follows. Given the current state, $q$, we propose a new state, $q^{\ast}$ using a symmetric proposal distribution such that $P(q^{\ast} | q) = P(q|q^{\ast})$. We then accept the proposed state as our new state with the following probability:
\begin{eqnarray*}
\min(1, \exp[{ U(q)-U(q^{\ast}) }])
\end{eqnarray*}

The standard random walk Metropolis generates proposals by sampling from a normal distribution with its mean set to the current state, $q$. There are more efficient strategies to generate proposals. Among these, Hamiltonian Monte Carlo (HMC) has become increasingly popular due to its capability of making distant proposals (i.e., low autocorrelation) with high acceptance probability. More specifically, HMC introduces a Hamiltonian dynamics system with auxiliary momentum variables $p$ to propose samples of $q$ in a Metropolis framework that explores the parameter space more efficiently compared to standard random walk proposals. Following the dynamics of the introduced Hamiltonian system, HMC generates proposals jointly for $(q, p)$ using the following system of differential equations:
\begin{align}
\frac{dq_i}{dt} &= \frac{\partial H}{\partial p_i}\label{eq:Hq}\\
\frac{dp_i}{dt} &= -\frac{\partial H}{\partial q_i}\label{eq:Hp}
\end{align}
where the Hamiltonian function is defined as $H(q,p) = U(q) + \frac12p^TM^{-1}p$. The quadratic kinetic energy function $K(p) = \frac12p^TM^{-1}p$ corresponds to the negative log-density of a zero-mean multivariate Gaussian distribution with the covariance $M$. Here, $M$ is known as the mass matrix, which is often set to the identity matrix, $I$, but can be used to precondition the sampler using Fisher information \cite{girolami11}. By simulating the Hamiltonian dynamics system together with the correction (i.e., accept/reject) step, HMC generates samples from a joint distribution of $(q,p)$ defined by 
\[
P(q,p) \propto \exp\left(-U(q)-\frac12p^TM^{-1}p\right)
\]
Notice that $q$ and $p$ are independent in general.

Each sample from the HMC algorithm is generated by two steps: the proposal step and the correction step. In the proposal step, new values for the momentum variable $p$ are drawn from their Gaussian distribution. Starting from the current state $(q,p)$, the Hamiltonian dynamics system \eqref{eq:Hq},\eqref{eq:Hp} is simulated for $L$ steps using the leapfrog method (Algorithm \ref{alg:HMC}), with a stepsize of $\epsilon$. Here, $L$ and $\epsilon$ are parameters which needs to be tuned to obtain a reasonable acceptance probability. In the correction step, the proposed state $(q^{\ast},p^{\ast})$ at the end of the trajectory is accepted as the next state of the Markov chain with probability $\min(1,\exp[-H(q^{\ast},p^{\ast}) + H(q,p)])$ and the position variable $q$ is updated correspondingly. These steps are presented in Algorithm \ref{alg:HMC}.

\begin{algorithm}[t]
\KwIn{Starting position $q^{(1)}$ and step size $\epsilon$}
 \For{$t =1,2,\cdots$ }{
  \textit{Resample momentum $p$}\\
  $p^{(t)} \sim \mathcal{N}(0,M),\;(q_0,p_0)$ = $(q^{(t)},p^{(t)})$\\
  \textit{Simulate discretization of Hamiltonian dynamics:}\\
  \For{$l = 1$ to $L$} {
  $p_{l-1} \leftarrow p_{l-1} - \frac{\epsilon}{2} \frac{\partial U}{\partial q}(q_{l-1})$\\
  $q_l \leftarrow q_{l-1} + \epsilon M^{-1}p_{l-1}$\\
   $p_l \leftarrow p_l - \frac{\epsilon}{2} \frac{\partial U}{\partial q}(q_{l})$
  }
  $(q^{\ast},p^{\ast}) = (q_L,p_L)$\\
  \textit{Metropolis-Hasting correction:} \\
  $u \sim \text{Uniform}[0,1]$\\
  $\rho = \exp[{ H(q^{(t)},p^{(t)})-H(q^{\ast},p^{\ast}) }]$\\
  \lIf{$u < \min(1,\rho)$,} {$q^{(t+1)} = q^{\ast}$}
 }
\caption{Hamiltonian Monte Carlo}
\label{alg:HMC}
\end{algorithm}
\vspace{5pt}
Note that in Algorithm \ref{alg:HMC}, when simulating the Hamiltonian dynamics system, we need to repeatedly compute the gradient of the potential energy function $U$. This could be extremely time consuming. Many attempts have been made in recent years to reduce this cost in order to improve the overall computational efficiency of HMC. See for example, Shahbaba et~al.~\cite{shahbabaSplitHMC} and Chen \cite{chen14}. In this paper, we claim that the smooth dependence on parameters can be exploited to approximate the gradient effectively using precomputed or previous-computed values on a grid of points. We will discuss our method in details in the following section. 

\section{Precomputing Strategies}\label{precompute}
\subsection{Insights from an Illustrative Example}
We start with a simple example to motivate our approach. Consider a bivariate Gaussian distribution with known covariance matrix and a conjugate prior
\[
Y|\mu \sim \mathcal{N}(\mu,\Sigma),\quad \mu \sim \mathcal{N}(\mu_0,\Sigma_0)
\]
Note that in this case, the posterior distribution has a closed form so MCMC is not required. However, we use this example to motivate our method. For this problem, the potential energy function and its gradient are given by
\begin{align}
U(\mu) &= \frac12\sum_{i=1}^{N}(y_n-\mu)^T\Sigma^{-1}(y_n-\mu) + \frac12(\mu-\mu_0)^T\Sigma_0^{-1}(\mu-\mu_0)\\
\frac{\partial U}{\partial \mu} &= N\Sigma^{-1}(\mu-\bar{Y}) + \Sigma_0^{-1}(\mu-\mu_0),\quad \bar{Y} = \sum_{i=1}^{N}y_i/N\label{eq:grad}
\end{align}
In the gradient function \eqref{eq:grad}, all the information about the parameter is contained in one single value $\bar{Y}$ (i.e., the \emph{sufficient statistic} for $\mu$). Therefore, if we precompute $\bar{Y}$, gradient computation of the potential energy function $U$ could be reduced to a simple matrix vector multiplication.  Moreover, the gradient function itself is a linear function. In this 2D case, the essential information of $\frac{\partial{U}}{\partial \mu}$ can be captured by its function values at three non-collinear points (left panel of Fig.\ref{fig:2dex}). On the other hand, samples from the posterior distribution are concentrated around the high density region where the neighborhood of one sample is frequently visited in the simulations of Hamiltonian dynamics (right panel of Fig.\ref{fig:2dex}). We use these insights to develop a method that can approximate the gradient function using precomputed values in order to accelerate standard HMC.

\begin{figure}
\begin{tabular}{cc}
\includegraphics[width=0.45\textwidth, height=0.3\textwidth]{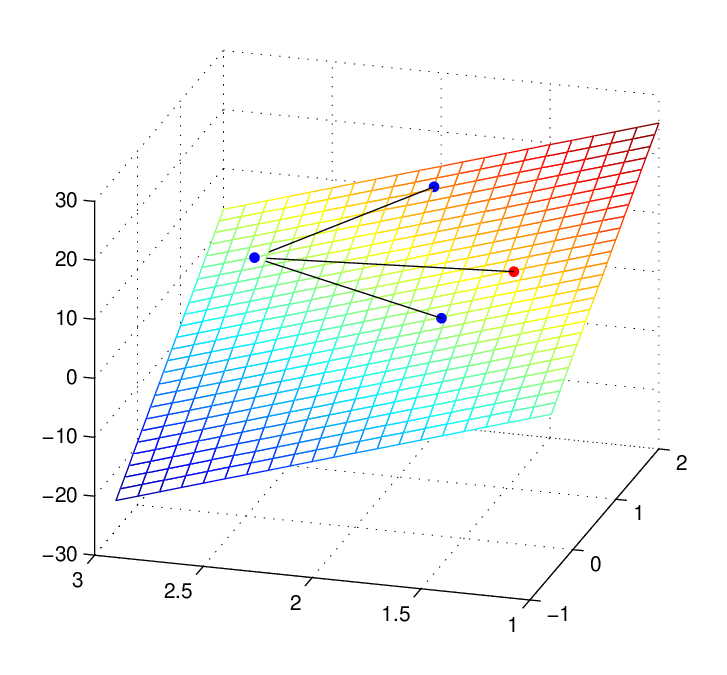}&
\includegraphics[width=0.45\textwidth, height=0.3\textwidth]{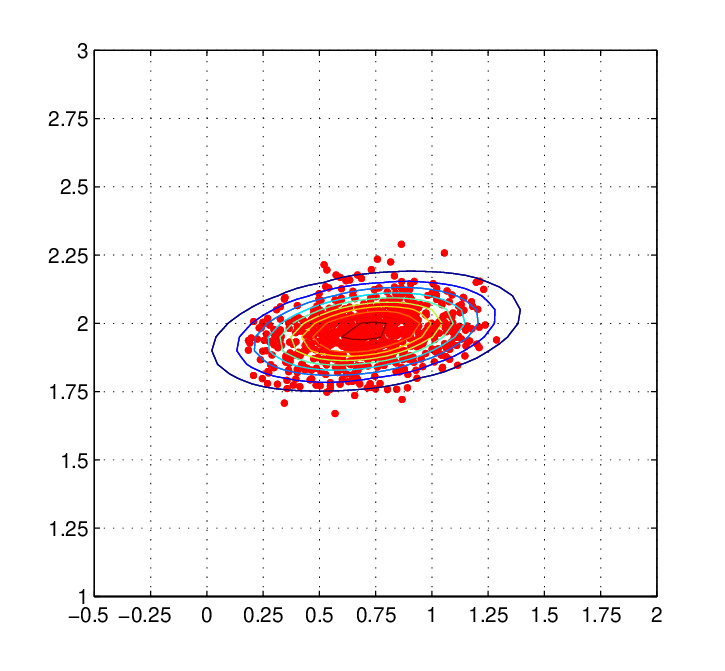} \\
(a) $\nabla_{\mu}U$  & (b)  HMC samples 
\end{tabular}
\caption{2D Gaussian example: (a) the graph of the first component of $\nabla_{\mu}U$. The function value at red point can be obtained by interpolation when the function values at three blue points are known. (b) HMC samples from the posterior distribution.}\label{fig:2dex}
\end{figure}

\subsection{Force Approximation}
If we could solve Hamilton's equations (\ref{eq:Hq} and \ref{eq:Hp}) analytically, the acceptance probability of new proposals in HMC would be exactly one (i.e., each proposal is accepted) because of the conservation of the Hamiltonian \cite{neal11}. However, since solving these equations exactly is too hard in practice, we usually approximate them by discretizing time and using the leapfrog method (Algorithm \ref{alg:HMC}). As a result, the acceptance probability may be less than one. The tradeoff between the accuracy of the proposal-generating mechanisms and Metropolis acceptance probability can go beyond time discretization. Therefore, in this paper we ask the following question: \emph{can we properly approximate the proposal-generating mechanism in order to reduce computational complexity while keeping the acceptance probability at a reasonable level?} We answer this question in the remaining part of this section. 

Note that we can rewrite Hamilton's equations as follows:
\begin{align}
\frac{dq_i}{dt} &= [M^{-1}p]_i\\
\frac{dp_i}{dt} &= -\frac{\partial U}{\partial q_i}\label{eq:force}
\end{align}
The routine of the trajectory for one proposal step will be determined by both the random initialization of the momentum and the negative gradient of the potential energy function, which is called \emph{force} in Physics,
\[F=-\frac{\partial U}{\partial q}\]
The random momentum enables the scheme to explore the target distribution stochastically, and the fictitious force guides the sampler in the right direction so that the entire sampling method would be more efficient than random walk proposals. However, the computation of the true force $F$ is quite expensive. To alleviate this issue, we propose to construct a Hamiltonian dynamics system, at this time for the \emph{proposal step} only, using an alternative Hamiltonian function,
\[
\tilde{H}(q,p) = \tilde{U}(q) + K(p)
\]
where $\tilde{U}$ is an approximation to the potential energy $U$, whose corresponding negative gradient $\tilde{F}$ (which is an approximation to the true force function $F$) can be computed relatively fast. Note that the alternative Hamiltonian $\tilde{H}$ induces a dynamics system which is also reversible and volume preserving, the convergence to the correct target distribution can be guaranteed if we use the original Hamiltonian when calculating the acceptance probability of the proposals (see Appendix \ref{sec:correct} for more details). Since the simulation of a Hamiltonian dynamics system only involves the force function, it suffices to find an approximate force function, $\tilde{F}$, directly.

\subsection{Naive Grid HMC}
To approximate the force function, one could simply use a piecewise constant function, which corresponds to a piecewise linear approximation of the potential energy. In most cases, the high density region of the posterior distribution can be covered by a finite domain $D$, henceforth called ``domain of interest''. If we partition $D$ with a fine grid, justification of an appropriate piecewise constant approximation to the force function $F$ is guaranteed by the smooth dependence of $F$ (or $U$) in parameter space. For a 2-dimensional problem, suppose our domain of interest is $D = [a,b]\times[c,d]$. Given the grid points
\[
x_i = a + i\Delta x,\quad y_j = c + j\Delta y, \quad i,j = 0,1,\ldots,N_p
\]
where $\Delta x,\;\Delta y$ are the corresponding grid sizes, for each cell, $C_{i,j} = [x_{i-1},x_i]\times [y_{j-1},y_j]$, we approximate the force function by its value at the center of the grid, $c_{i,j} = (x_{i-1/2},y_{j-1/2})$: 
\[
\tilde{F}(q) = F_{i,j} \buildrel\triangle\over = F(c_{i,j}), \text{ if }  q \in C_{i,j}
\]  
By the smoothness of $F$, $\|\tilde{F}-F\|_{\infty} \rightarrow 0$ as $\Delta x, \Delta y \rightarrow 0$. Therefore, we can always find some fine grid to achieve the desired approximation accuracy.

\begin{figure}[t]
\begin{center}
\includegraphics[width = .7\textwidth]{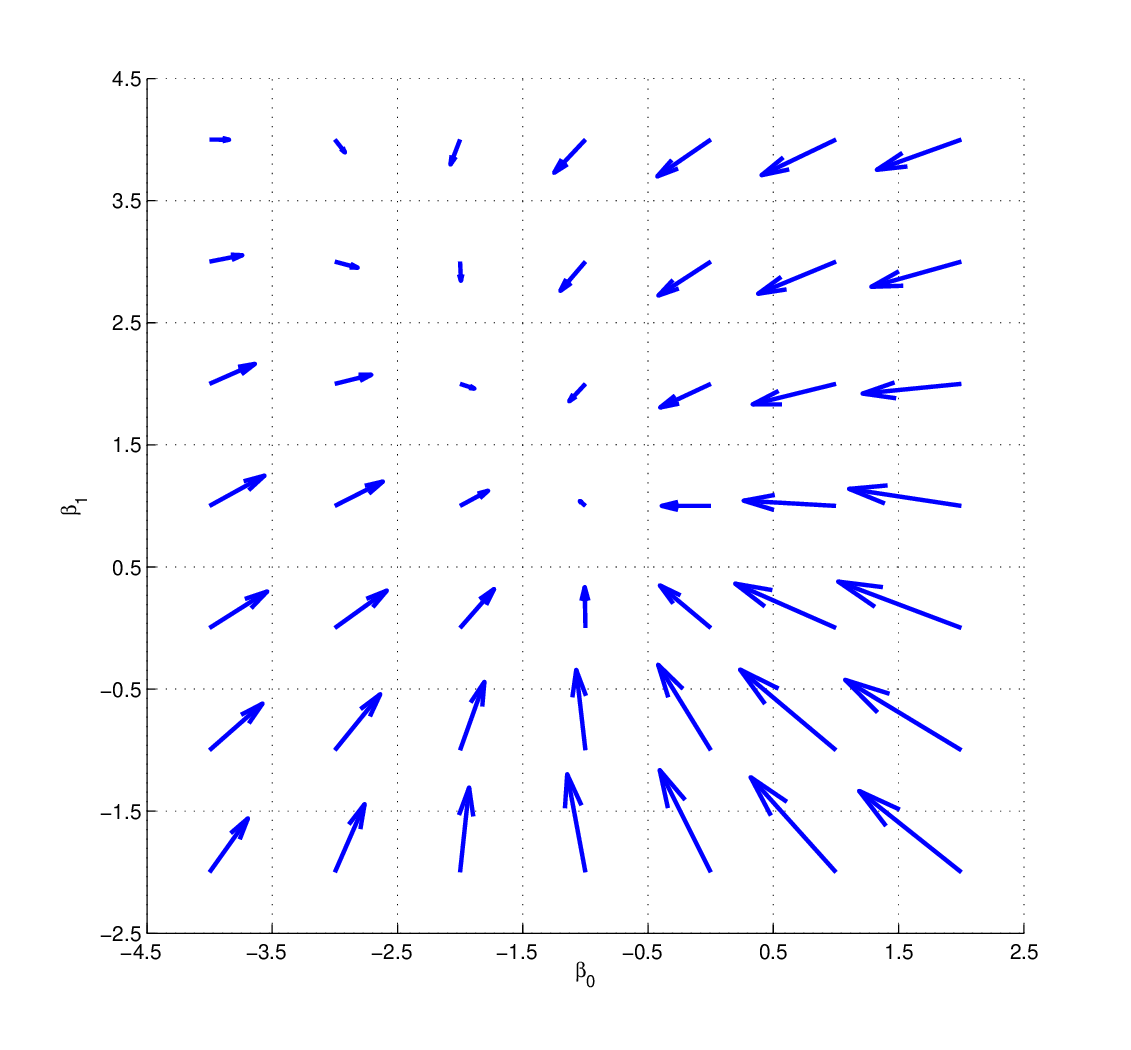}
\caption{Force map of a logistic regression model}\label{fig:forcemap}
\end{center}
\end{figure}

Figure \ref{fig:forcemap} shows a piecewise constant approximation to the force function of a logistic regression model with design matrix $X =(\mathbf{1},X_1)$ and true parameter $\beta = (-1,1)^T$, where $X_1$ follows standard normal distribution. The binary responses $Y=(y_1,y_2,\ldots,y_N)^T$ are sampled independently from Bernoulli distributions
\[
y_i \sim \mathrm{Bernoulli}(p_i),\quad p_i = \frac{\exp(x_i\beta)}{1+\exp(x_i\beta)}
\]
Therefore, the likelihood function is 
\[
L(\beta|X,Y) \propto \prod_{i=1}^{N}p_i^{y_i}(1-p_i)^{1-y_i}
\]
and the potential energy function and the force function are
\begin{align}
U(\beta) &=  -\sum_{i=1}^N\big[y_ix_i\beta -\log(1+\exp(x_i\beta))\big]\\
F(\beta) &= -\frac{\partial U}{\partial\beta} = X^T(Y-P)
\end{align}
where $P = (p_1,p_2,\ldots,p_N)^T$. It can be seen from the graph that: (i) the approximate force ``map'' does point to the right direction so that it provides valid geometric information for HMC; (ii) the approximate force function also changes smoothly, which means that the numerical stability of the leap-frog scheme can be maintained with approximately the same step size as standard HMC. As a result, the proposed scheme with piecewise constant force functions would be consistent and stable. Therefore, we can precompute the piecewise constant function $\tilde{F}$ in advance. When evaluating the force function in the simulation of the Hamiltonian dynamics system, we locate the cell $(i,j)$ for the current parameter $q$ and read the approximate function value $\tilde{F}(q)$ from the precomputed force map. We summarize this approach in Algorithm \ref{HMC:gridmesh} and refer to it as Grid HMC (GHMC).

Our initial results showed that the Naive Grid HMC method would work well for simple problems. However, implementation of this method in general involves two challenges. First, its extension to high dimensional problems could be problematic because as the number of parameters increases, the number of grid nodes at which we need to evaluate the approximate force map grows exponentially. In other words, the method will encounter the curse of dimensionality. The second challenge is related to finding the domain of interest. We will address those two issues in the following two subsections respectively.

\begin{algorithm}[t]
\KwIn{Starting position $q^{(1)}$ and step size $\epsilon$}
Precompute the approximate force map $\tilde{F}: F_{i,j} = -\left(\frac{\partial{U}}{\partial q}\right)_{i,j}$\\
 \For{$t =1,2,\cdots$ }{
  \textit{Resample momentum $p$}\\
  $p^{(t)} \sim \mathcal{N}(0,M),\;(q_0,p_0)$ = $(q^{(t)},p^{(t)})$\\
  \textit{Simulate discretization of Hamiltonian dynamics:}\\
  Find the position of $q_0$ in the force map: $(i_0,j_0)$\\
  \For{$l = 1$ to $L$} {
  $p_{l-1} \leftarrow p_{l-1} + \frac{\epsilon}{2} F_{i_{l-1},j_{l-1}}$\\
  $q_l \leftarrow q_{l-1} + \epsilon M^{-1}p_{l-1}$\\
   Find the position of $q_l$ in the force map: $(i_l,j_l)$\\
    $p_l \leftarrow p_l + \frac{\epsilon}{2} F_{i_l,j_l}$
  }
  $(q^\ast,p^\ast) = (q_L,p_L)$\\
  \textit{Metropolis-Hasting correction:} \\
  $u \sim \text{Uniform}[0,1]$\\
  $\rho = e^{ H(q^{(t)},p^{(t)})-H(q^\ast,p^\ast) }$\\
  \lIf{$u < \min(1,\rho)$,} {$q^{(t+1)} = q^\ast$}
 }
\caption{Naive Grid HMC}
\label{HMC:gridmesh}
\end{algorithm}
\vspace{10pt}

\subsection{Sparse Grid HMC} The sparse grid interpolation method use a special discretization technique to approximate a smooth function over a sparse grid of points (\cite{bung04, klimke05, barth00}). More specifically, it uses a hierarchical basis (a representation of a discrete function space that is equivalent to the conventional nodal basis) and a sparse tensor product construction. Discretization on sparse grids employs $\mathcal{O}(N\cdot \log(N)^{d-1})$ grid points only, where $d$ denotes the dimension and $N$ denotes the number of grid points at the boundary in each coordinate direction (i.e., the mesh size is $h = 1/N$). Using piecewise linear basis functions, the interpolation accuracy could be $\mathcal{O}(N^{-2}\cdot(\log N)^{d-1})$ with respect to the $L_2$ norm and $L_\infty$ norm if the solution has bounded second mixed derivatives. Note that for a full grid, $N^d$ grid points are needed to achieve an approximation accuracy of $\mathcal{O}(N^{-2})$. 

\begin{figure}[t]
\begin{tabular}{cccc}
\includegraphics[width = .23\textwidth]{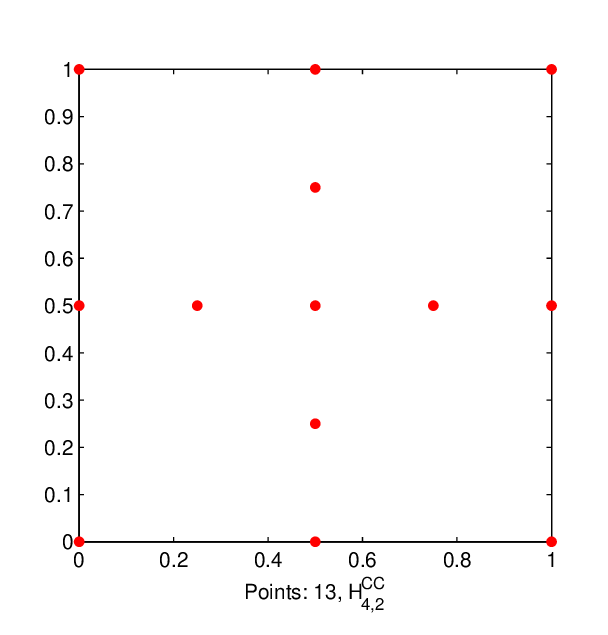} &
\includegraphics[width = .23\textwidth]{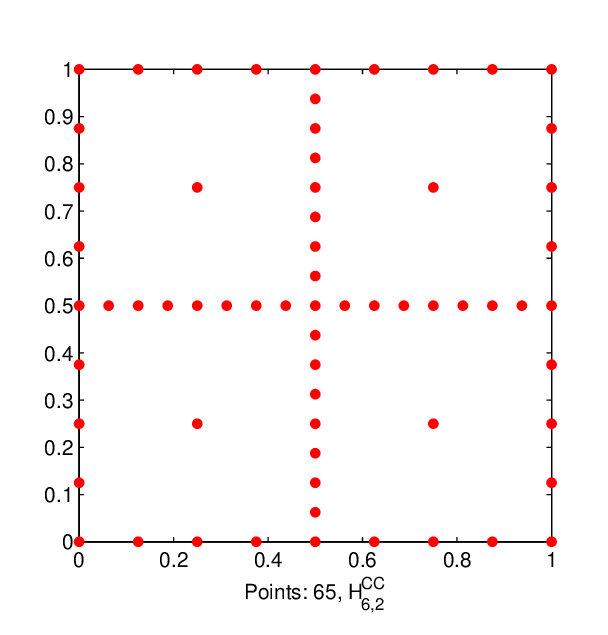} &
\includegraphics[width = .23\textwidth]{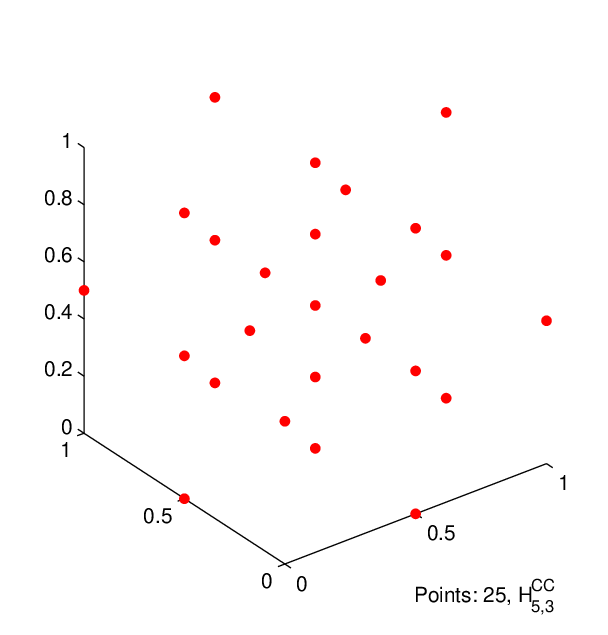} &
\includegraphics[width = .23\textwidth]{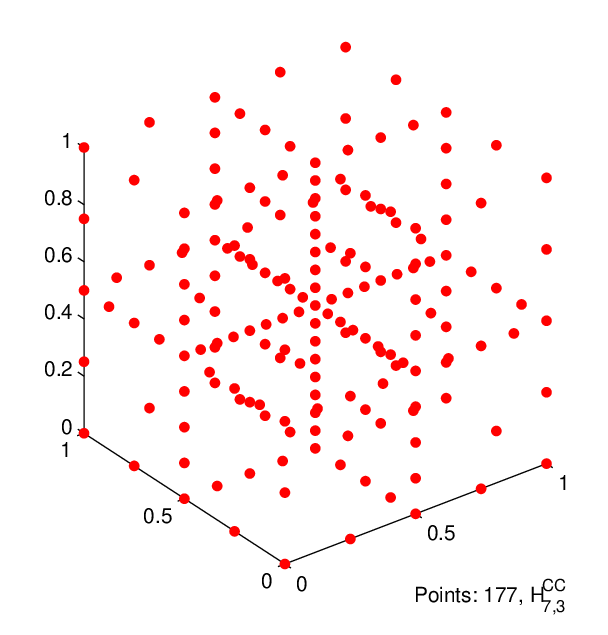} \\
$H^{CC}_{4,2}$  &  $H^{CC}_{6,2}$ & $H^{CC}_{5,3}$ & $H^{CC}_{7,3}$
\end{tabular}
\caption{Clenshaw-Curtis type sparse grids}\label{fig:CC}
\end{figure}

\subsubsection{Smolyak's formula} 
Assume that we want to approximate the smooth functions $f:[0,1]^d\rightarrow \mathbb{R}$ using a finite number of  function values (at support nodes). For the one dimensional case, the interpolation formula is given by 
\[
U^i(f) = \sum_{j=1}^{m_i}f(x^i_j)\cdot a^i_j
\]
where $i\in\mathbb{N},\;X^i=\{x_j^i\in [0,1]|j=1,\ldots,m_i\}$ are the support nodes, and $a_j^i\in C([0,1])$ are the basis functions. We could use the following tensor product for multidimensional cases:
\begin{equation}\label{eq:tensor}
(U^{i_1}\otimes\cdots\otimes U^{i_d}) = \sum_{j_1=1}^{m_{i_1}}\cdots\sum_{j_d=1}^{m_{i_d}}f(x_{j_1}^{i_1},\ldots,x_{j_d}^{i_d})\cdot (a_{j_1}^{i_1}\otimes\cdots\otimes a_{j_d}^{i_d})
\end{equation}
However, the above product formula requires a large number ($m_{i_1}\cdots m_{i_d}$) of support nodes, which are sampled on the full grid. Smolyak's formula then can be applied here to reduce the number of support nodes while maintaining the approximation quality of the interpolation formula up to a logarithmic factor. With $U^0 = 0$, define
\[
\Delta^i = U^i-U^{i-1}, \quad \forall\; i \in \mathbb{N}
\]
Moreover, we put $|\vect{i}| = i_1+\cdots + i_d$ for $\vect{i} \in \mathbb{N}^d$. Then Smolyak's algorithm is given by
\begin{equation}
A_{q,d}(f) = \sum_{|\vect{i}|\leq q} (\Delta^{i_1}\otimes\cdots\otimes\Delta^{i_d})(f) = A_{q-1,d}(f) + \underbrace{\sum_{|\vect{i}|=q}(\Delta^{i_1}\otimes\cdots\otimes\Delta^{i_d})(f)}_{\Delta A_{q,d}(f)}\label{eq:smolyak}
\end{equation}
for integers $q\geq d$, where $A_{d-1,d} = 0$. In fact, \eqref{eq:smolyak} can be presented in terms of the univariate interpolation formulas \cite{wasi95}, 
\[
A_{q,d}(f) = \sum_{q-d+1\leq |\vect{i}| \leq q}\;(-1)^{q-|\vect{i}|}\cdot {d-1\choose q-|\vect{i}|}\cdot (U^{i_1}\otimes\cdots\otimes U^{i_d})(f)
\] 
Therefore, only the function values at the sparse grid 
\begin{equation}\label{eq:sparsegrid}
H_{q,d} = \bigcup_{q-d+1\leq |\vect{i}| \leq q} (X^{i_1}\times\cdots\times X^{i_d})
\end{equation}
are needed to evaluate $A_{q,d}(f)$ . It is better to select the sets $X^i$ in a nested fashion ($X^i \subset X^{i+1}$) to obtain many recurring points with increasing $q$. 

\subsubsection{Sparse grid and Multivariate hierarchical structure}\label{sec:sparsegrid}
There are many possibilities to construct nested sparse grids. As an example, Fig.~\ref{fig:CC} shows the Clenshaw-Curtis type sparse grids $H^{CC}$ in two and three dimensional spaces. With appropriate sparse grid and basis functions $a$, the multivariate interpolation formula~\eqref{eq:smolyak} can be implemented in a hierarchical form where

\begin{equation}\label{eq:multhier}
\Delta A_{q,d}(f) = \sum_{|\vect{i}|=q}\sum_{x_{j_1}^{i_1}\in X_{\Delta}^{i_1}}\cdots\sum_{x_{j_d}^{i_d}\in X_{\Delta}^{i_d}}\big(f(x_{j_1}^{i_1},\ldots,x_{j_d}^{i_d})-A_{q-1,d}(f)(x_{j_1}^{i_1},\ldots,x_{j_d})\big)\cdot(a_{j_1}^{i_1}\otimes\cdots\otimes a_{j_d}^{i_d})
\end{equation}

\begin{algorithm}[t]
\KwIn{Starting position $q^{(1)}$ and step size $\epsilon$}
Precompute the hierarchical surpluses for Smolyak's formula $A_{k+d,d}$ of potential energy $U$\\
 \For{$t =1,2,\cdots$ }{
  \textit{Resample momentum $p$}\\
  $p^{(t)} \sim \mathcal{N}(0,M)$\\
  $(q_0,p_0)$ = $(q^{(t)},p^{(t)})$\\
  \textit{Simulate discretization of Hamiltonian dynamics}\\
  \For{$l = 1$ to $L$} {
  $p_{l-1} \leftarrow p_{l-1} - \frac{\epsilon}{2} \nabla A_{k+d,d}(U)(q_{l-1})$\\
  $q_l \leftarrow q_{l-1} + \epsilon M^{-1}p_{l-1}$\\
  $p_l \leftarrow p_l - \frac{\epsilon}{2} \nabla A_{k+d,d}(U)(q_l)$
  }
  $(q^\ast,p^\ast) = (q_L,p_L)$\\
   \textit{Metropolis-Hasting correction:} \\
  $u \sim \text{Uniform}[0,1]$\\
  $\rho = e^{ H(q^{(t)},p^{(t)})-H(q^\ast,p^\ast) }$\\
  \lIf{$u < \min(1,\rho)$,} {$q^{(t+1)} = q^\ast$} }
\caption{Sparse Grid HMC}
\label{HMC:SG}
\end{algorithm} 
\vspace{5pt}

The hierarchical surpluses 
\[
w_{\vect{j}}^{k,\vect{i}} \buildrel\triangle\over = f(\vect{x}_{\vect{j}}^{\vect{i}}) - A_{k+d-1,d}(\vect{x}_{\vect{j}}^{\vect{i}})
\]
introduced by Bungartz \cite{bung98} can be used to obtain an estimate of the current approximation error and terminate the algorithm automatically when a desired accuracy is reached. More detailed information about the construction of sparse grid and basis functions and the derivation of the hierarchical form is provided in the appendix.

Using sparse grid interpolation \eqref{eq:smolyak} based on the Smolyak algorithm, we can generalize our Naive Grid HMC method to relatively higher dimensional problems. The hierarchical surpluses for the energy function $U$ can be precomputed with certain type of sparse grid and $\nabla A_{q,d}$ can be called to replace the gradient computation (see Algorithm \ref{HMC:SG}).

\subsection{Domain of Interest}\label{subsec:ROI}
The grid needs to be specified over a finite domain of interest such that there is a good balance between the cost and efficiency of the precomputing strategy. That is, we need to find an appropriate bounded domain $D$ that covers most of the high density areas without creating cells that are rarely visited by the sampler. 

Note that for points outside of $D$, one can still use the standard HMC method; that is, there will not be any computational saving for these points. More specifically, the overall potential energy function can be presented as follows: 
\begin{eqnarray} \label{unifiedPotential}
\hat{U}(q) = U(q)(1-\mathbbm{1}_D) + \tilde{U}\mathbbm{1}_D
\end{eqnarray}
where $U$ is the original potential, $\tilde{U}$ is the precomputed approximation over $D$, and $\mathbbm{1}$ is the indicator function. Therefore, within the identified domain the energy function is approximated; whereas, the energy function remains exact outside of the domain. Note that the proposal generating mechanism remains symmetric since after sampling the momentum variable the trajectory is deterministic; therefore, by reversing the time and negating the momentum at the end of the trajectory (i.e., proposal) we always come back to the starting point (i.e, current state of MCMC). The only difference with the standard HMC is that some parts of the trajectory may follow precomputed directions. Our proposed algorithm still uses the original Hamiltonian to compute the acceptance probability, whether the trajectory is all based on approximate gradients (i.e., the sampler remains in the domain), exact gradients (i.e., the sampler remains outside of the domain), or partly approximate, partly exact (i.e., when the sampler moves inside or outside of the domain). Using \ref{unifiedPotential}, the leapfrog scheme used to discretize the corresponding Hamiltonian dynamics system is as follows:
\begin{align*}
p_i(t+\epsilon/2) &= p_i(t) - \epsilon/2\frac{\partial \hat{U}}{\partial q_i}(q(t))\\
q_i(t+\epsilon) &= q_i(t) + \epsilon \frac{p_i(t+\epsilon)}{m_i} \\
p_i(t+\epsilon) &= p_i(t+\epsilon/2) - \epsilon/2\frac{\partial \hat{U}}{\partial q_i}(q(t+\epsilon))
\end{align*}
Note that the scheme is also time reversible and volume preserving (each of the equations are shear transformations). Therefore, simulating the above induced dynamics system and using the original Hamiltonian in the computation of acceptance probability guarantees the convergence to the correct target distribution (see Appendix A). Finally, it is easy to show that the chain remains ergodic and can move between the two domains. To see this, notice that within the first step of the leapfrog, $q$ is updated as follows:
\begin{eqnarray*}
q^{*} &  \sim & N(q - \frac{\epsilon^2}{2}\nabla_{q} \hat{U}(q), \epsilon^2 M)
\end{eqnarray*} 
where the support is $\mathbb{R}^{d}$; therefore, the sampler has non-zero probability to move inside and outside of the domain.

To find the domain of interest, we use Laplace's approximation,
\[
q|Y \buildrel\cdot\over \thicksim
\mathcal{N}(\hat{q},\mathcal{J}^{-1}(\hat{q}))
\]
where $\hat{q}$ is the posterior mode which can be estimated using fast optimization methods, and $\mathcal{J}(\hat{q}) = H_U(\hat{q})$ is the Hessian matrix at the point. Given a pre-specified probability, $p$, we can find a domain with probability $p$ based on the above normal approximation. 

\begin{figure}
\begin{tabular}{cc}
\includegraphics[width=0.45\textwidth, height=0.3\textwidth]{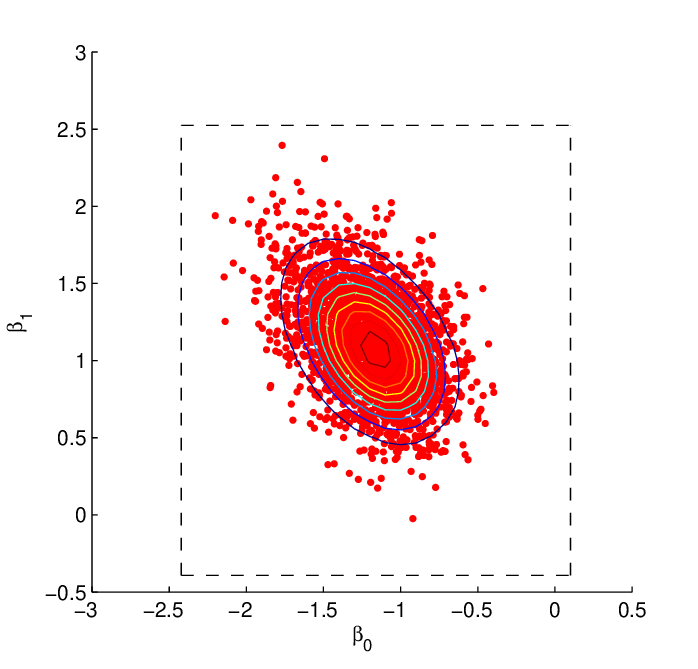}&
\includegraphics[width=0.45\textwidth, height=0.3\textwidth]{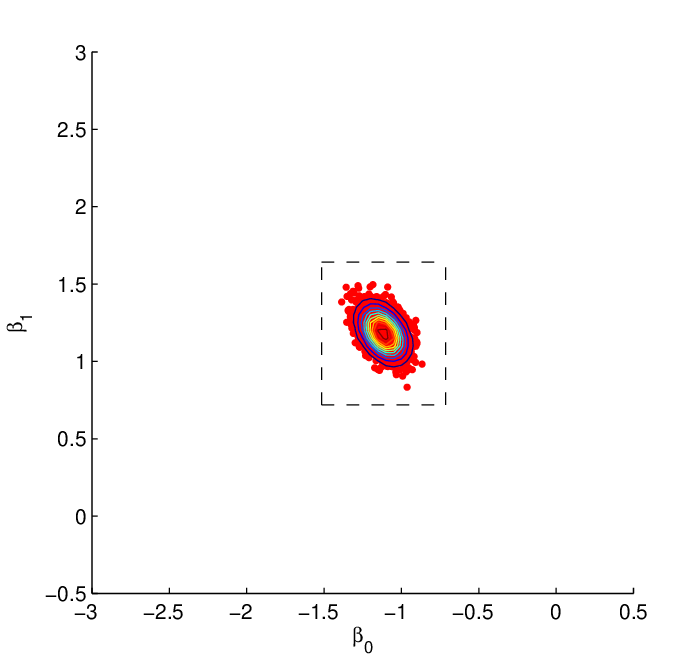} \\
(a) \small{$N=100$} &  (b) \small{$N=1000$} 
\end{tabular}
\caption{Domains of interest using Laplace's approximation for the logistic regression model.}\label{fig:ROI}
\end{figure}
Figure \ref{fig:ROI} shows the domain of interest $R$ for a logistic regression model. It can be seen that for different data sizes ($N=100 \text{ and }1000$), the corresponding domains of interest are adjusted automatically to capture the high density regions of the posterior distribution.

When the high density region is irregular and can not be represented well by a rectangular box, this might not be an efficient approach. Later, we will discuss a more general approach for such cases. 

\section{Experiments}\label{experiments}
In this section, we compare our proposed method to standard HMC using several experiments in terms of sampling efficiency. We define sampling efficiency as time-normalized effective sample size (ESS). Given $B$ MCMC samples for each parameter, we calculate the corresponding ESS =
$B[1 + 2\Sigma_{k=1}^{K}\gamma(k)]^{-1}$, where $\Sigma_{k=1}^{K}\gamma(k)$ is
the sum of $K$ monotone sample autocorrelations \cite{geyer92}. We use the
minimum ESS over all parameters normalized by the CPU time, $s$ (in seconds), as the overall measure
of efficiency: $\min(\textrm{ESS})/\textrm{s}$. The sparse grid interpolation is implemented using Matlab package $\mathtt{spinterp}$ \cite{klimke05}.

Empirical results show that both GHMC and Sparse Grid HMC (sgHMC) provide substantial improvement over standard HMC in terms of efficiency while maintaining relatively high acceptance rates. 

\subsection{Logistic regression} 
\begin{figure}
\begin{tabular}{cc}
\includegraphics[width=0.45\textwidth, height=0.3\textwidth]{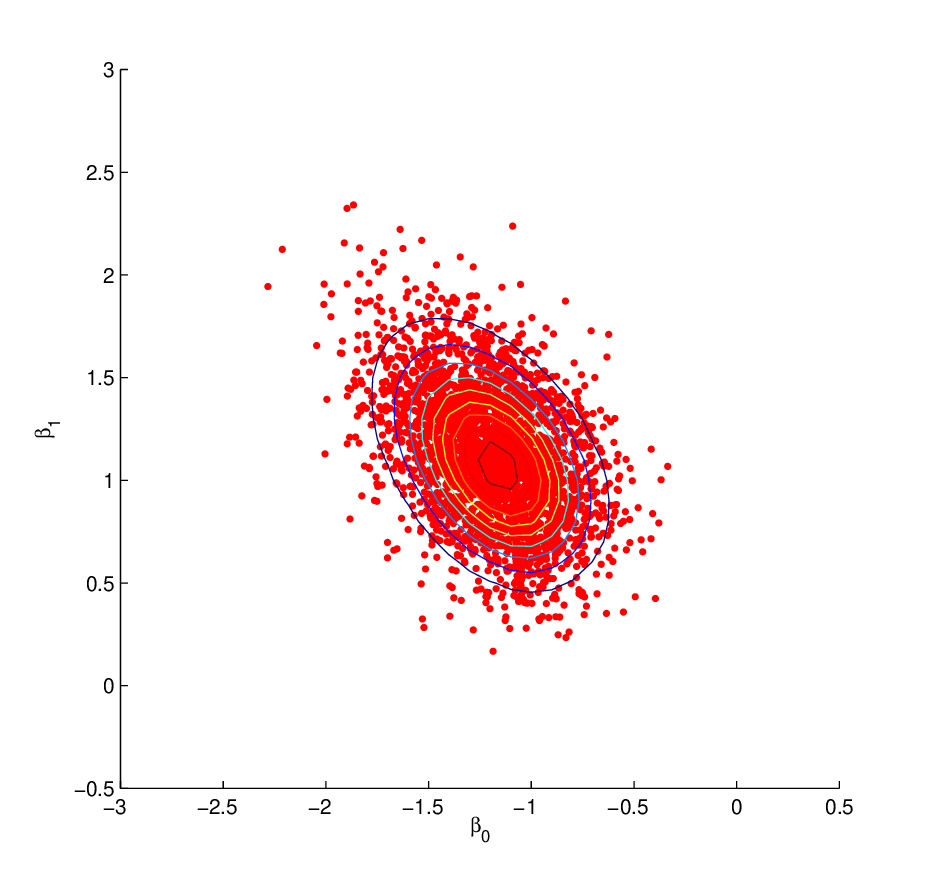}&
\includegraphics[width=0.45\textwidth, height=0.3\textwidth]{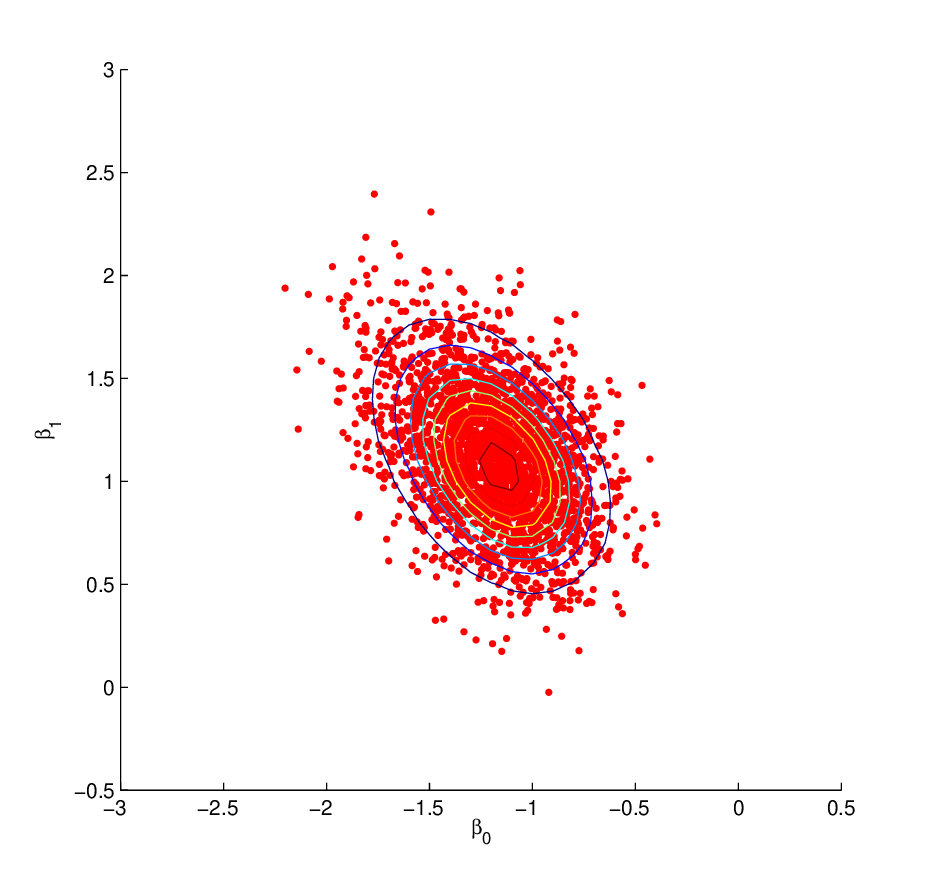}\\
HMC  &  GHMC  
\end{tabular}
\caption{HMC vs GHMC: logistic regression}\label{fig:clog}
\end{figure}

For our first example, we sample $N=100$ data points from a logistic regression model discussed in Section \ref{precompute} and choose the domain of interest to be $[-3,0.5]\times[-0.5,3]$ and set the grid size to $0.1$. Figure \ref{fig:clog} shows posterior samples using standard HMC and GHMC. Note that they both converge to the target distribution and explore the parameter space quite well. Table \ref{tab:clog} compares the performance of these algorithms based on 3200 MCMC iterations after burning the first 800 iterations. As we can see, GHMC outperforms standard HMC in terms of time-normalized ESS.

\begin{table}[!htp]
\begin{center}
\caption{Comparing HMC with GHMC using a logistic regression model. For each method, we provide the acceptance rate (AR), the CPU time (s) for each iteration and the time-normalized ESS}\label{tab:clog}
\begin{tabular}{|c|c|c|c|c|}
\hline
Method &AR &  ESS($\beta_0,\;\beta_1$)  & s/Iteration & min ESS/s   \\\hline
HMC&$0.9225$    &  $(3200,3200)$  &  $7.0157E\text{-}4$   & $1425.3707$ \\
GHMC&$0.7981$     &  $(3200,3200)$  & $3.318E\text{-}4$  & $3013.9031$\\\hline
\end{tabular}
\end{center}
\end{table}

\subsection{Banana-shaped distribution}
The potential energy function for the logistic regression model is quite similar to a Gaussian distribution model, where the resulting force function is relatively smooth. To investigate GHMC's ability to explore the parameter space with a more complicated geometry, we construct a banana-shaped posterior distribution of $\beta =(\beta_1,\beta_2|y)$ based on the following model:
\begin{align*}
y|\beta& \sim \mathcal{N}(\beta_1+\beta_2^2,\sigma_y^2)\\
\beta & \sim \mathcal{N}(0,\sigma_\beta^2)
\end{align*}
The data $\{y_i\}_{i=1}^{100}$ are generated with $\beta_1+\beta_2^2 = 1,\;\sigma_y = 2,\;\sigma_\beta =1.$ The potential energy function is
\begin{equation}
U(\beta) = \sum_{i=1}^{N}\frac{(y_i-\beta_1-\beta_2^2)^2}{2\sigma_y^2} + \frac{\beta_1^2+\beta_2^2}{2\sigma_\beta^2}
\end{equation}
and the force function is
\begin{equation}
F(\beta) = -\frac{\partial U}{\partial \beta} = \frac{\sum_{i=1}^{N}(y_i-\beta_1-\beta_2^2)}{\sigma_y^2} \cdot \begin{pmatrix}1\\2\beta_2\end{pmatrix}-\frac{ \beta}{\sigma_\beta^2}
\end{equation}

\begin{figure}
\begin{tabular}{cc}
\includegraphics[width=0.45\textwidth, height=0.3\textwidth]{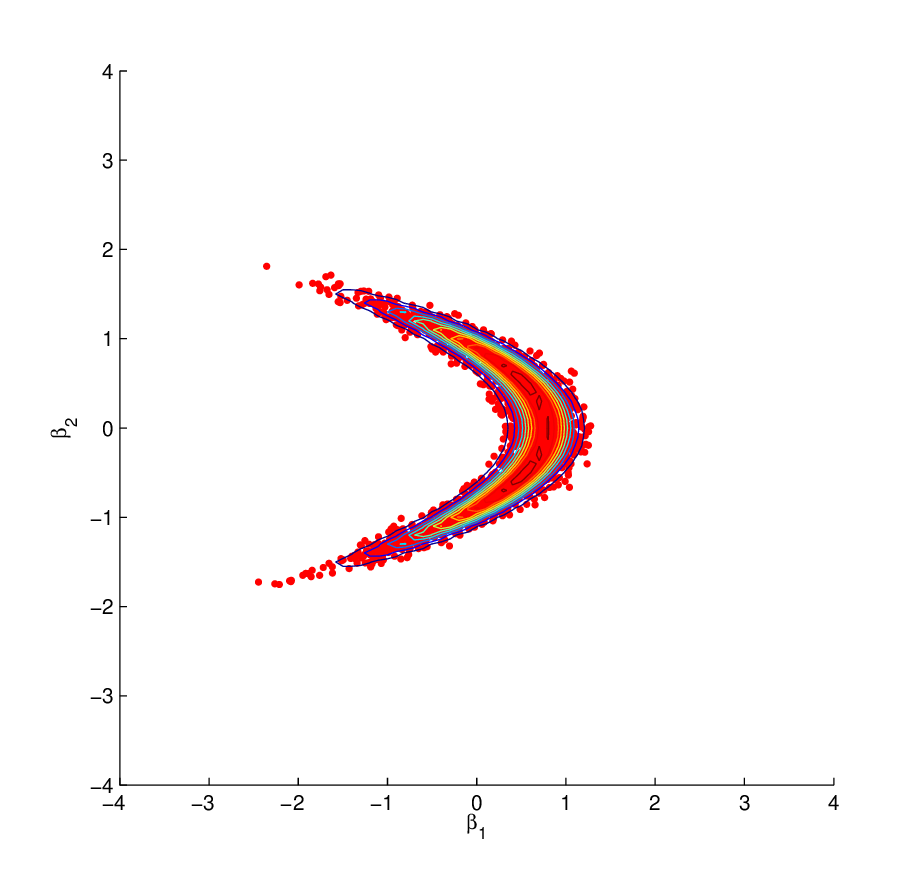}&
\includegraphics[width=0.45\textwidth, height=0.3\textwidth]{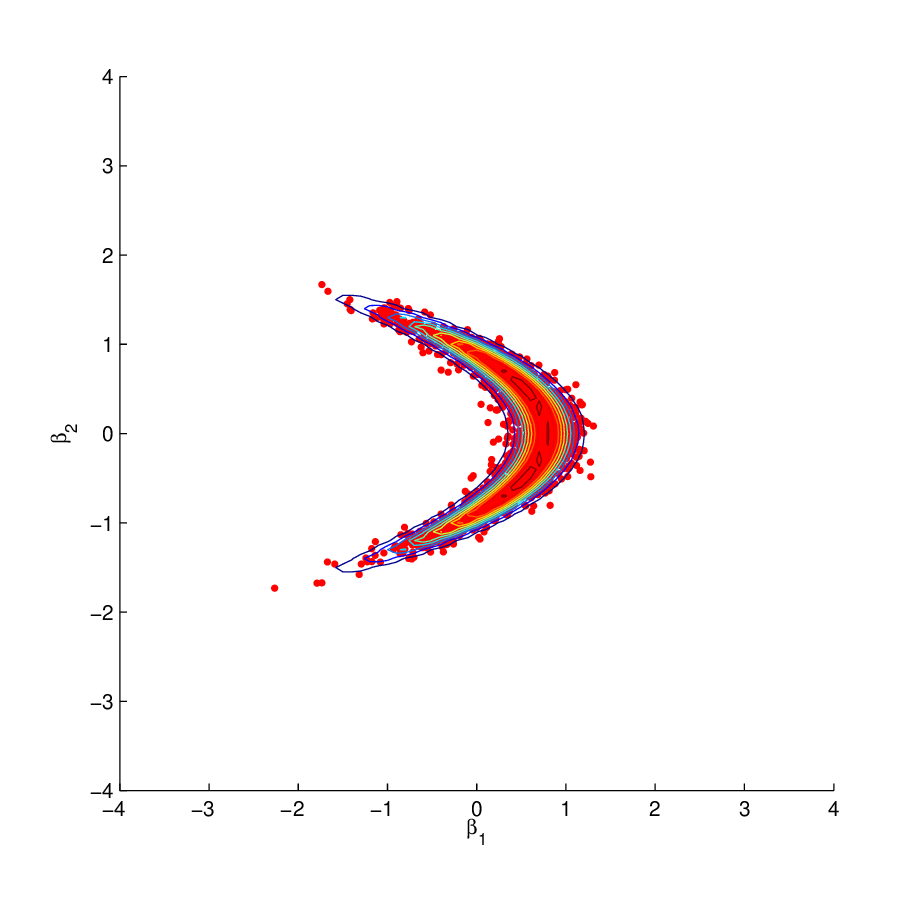} \\
HMC  &  GHMC  
\end{tabular}
\caption{HMC vs GHMC: banana-shaped distribution}\label{fig:cbanana}
\end{figure}

Here, we choose the domain of interest to be $[-4,4]\times[-4,4]$ and set grid size to $0.1$. Figure \ref{fig:cbanana} shows the samples for the posterior distribution using standard HMC and GHMC. As before, both methods converge to the target distribution and explore the parameter space quite well. Even though Banana-shaped distribution is more distorted and the force function is more complex, $\tilde{F}$ (grid size 0.1) still provide a good approximation to the true force function. Table \ref{tab:cbanana} compares the performance of these algorithms based on 3200 MCMC iterations after burning the first 800 iterations. As before, GHMC outperforms standard HMC in terms of time-normalized ESS.

\begin{table}[!htp]
\begin{center}
\caption{Comparing HMC with GHMC using a banana-shaped distribution model. For each method, we provide the acceptance rate (AR), the CPU time (s) for each iteration and the time-normalized ESS}\label{tab:cbanana}
\begin{tabular}{|c|c|c|c|c|}
\hline
Method &AR &  ESS($\beta_1,\;\beta_2$)  & s/Iteration & min ESS/s   \\\hline
HMC&$0.9353$   &  $(2403,1191.6)$  &  $3.8703E\text{-}4$   & $962.1346$ \\
GHMC&$0.6587$   &  $(893.8862,766.2423)$  & $1.4498E\text{-}4$  & $1651.5917$\\\hline
\end{tabular}
\end{center}
\end{table}

\subsection{Gaussian Process model}
For our third example, we use a Gaussian process model. Posterior sampling for these models tends to be quite difficult due to the computation cost associated with inverting the covariance matrix. See Neal \cite{nealGP98} and Rasmussen \cite{rasmussen96} for more details on Gaussian process. Here we construct a 2D Gaussian process with zero mean and the squared exponential covariance function,
\[
Y \sim \mathcal{N}(0,\Sigma),\quad \Sigma_{ij} = \eta\cdot\exp\left(-l\|x_i-x_j\|^2_2\right) + J\cdot\delta_{ij}
\]
where $\eta,l,J$ are positive hyperparameters with log-normal priors. 
\[
\log(\eta) \sim \mathcal{N}(-1,1),\quad \log(l) \sim \mathcal{N}(-1,1),\quad \log(J) \sim \mathcal{N}(-1,1)
\]
Let $\tilde{\eta} = \log(\eta),\;\tilde{l} = \log(l),\;\tilde{J} = \log(J)$, the potential energy function is 
\[
U(\tilde{\eta},\tilde{l},\tilde{J}) =  \frac12\log(|\Sigma|) + \frac12Y^T\Sigma^{-1}Y + \frac12\left[(\tilde{\eta}+1)^2+(\tilde{l}+1)^2+(\tilde{J}+1)^2\right]
\]
and the force function is
\[
F(\beta) = -\frac{\partial U}{\partial \beta} = \frac12\mathrm{tr}\left(\Sigma^{-1}\frac{\partial{\Sigma}}{\partial\beta}\right) -\frac12Y^T\Sigma^{-1}\frac{\partial{\Sigma}}{\partial\beta}\Sigma^{-1}Y + \beta + 1,\quad \beta = (\tilde{\eta},\tilde{l},\tilde{J})^T
\]

\begin{figure}
\begin{tabular}{cc}
\includegraphics[width=0.45\textwidth, height=0.3\textwidth]{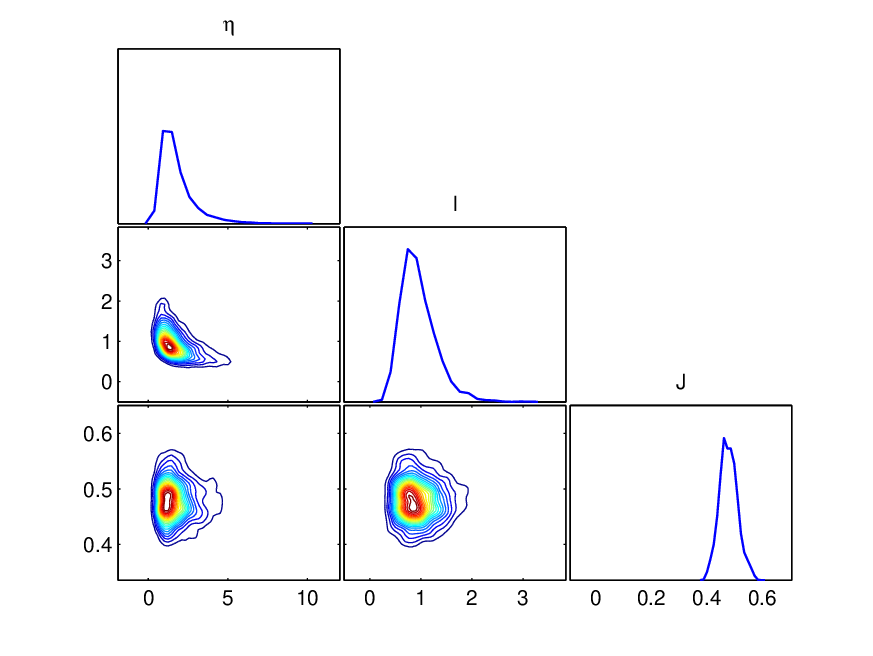}&
\includegraphics[width=0.45\textwidth, height=0.3\textwidth]{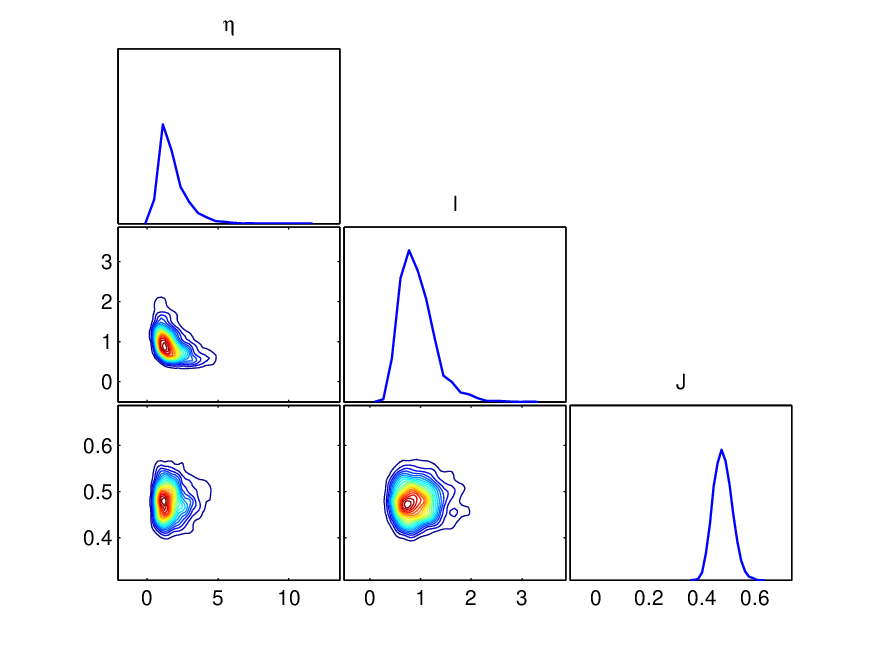} \\
HMC  &  sgHMC  
\end{tabular}
\caption{HMC vs sgHMC: Gaussian Process}\label{fig:gp}
\end{figure}

The domain of interest for $\beta=(\tilde{\eta},\tilde{l},\tilde{J})^T$ is set to be $[-1.6,1.6]\times[-1.6,1.6]\times[-1.2,0.4]$ where we train a sparse gird interpolator to replace the force function. Figure \ref{fig:gp} shows the samples from the posterior distribution given by standard HMC and sgHMC. Table.~\ref{tab:gp} compares the performance of the two algorithms based on 3200 MCMC iterations after 800 burn-in iterations.  As we can see, sgHMC substantially outperform standard HMC. 
\begin{table}[!htp]
\begin{center}
\caption{Comparing HMC with sgHMC using a Gaussian process model. For each method, we provide the acceptance rate (AR), the CPU time (s) for each iteration and the time-normalized ESS}\label{tab:gp}
\begin{tabular}{|c|c|c|c|c|}
\hline
Method &AR &  ESS($\eta,\;l,\;J$)  & s/Iteration & min ESS/s   \\\hline
HMC&$0.9472$   &  $(1021.7,1784.8,3200)$  &  $2.3547E\text{-}1$   & $1.3559$ \\
sgHMC&$0.7066$   &  $(828.7,1380.0,3200)$  & $2.9851E\text{-}2$  & $8.6752$\\\hline
\end{tabular}
\end{center}
\end{table}

\subsection{Elliptic PDE Inverse Problem} 
Our last example is a canonical inverse problem involving inference of the diffusion coefficient in an elliptic PDE (\cite{dashti11, conard14}). The forward model is to solve a two dimensional elliptic PDE 
\begin{equation}\label{eq:ePDE}
\nabla_{\vect{x}}\cdot(c(\vect{x},\theta)\nabla_{\vect{x}}u(\vect{x},\theta)) = 0
\end{equation}
where $\vect{x} =(x_1,x_2)\in[0,1]^2$ is the spatial coordinate. The boundary conditions are
\[
u(\vect{x},\theta)|_{x_2=0} = x_1,\quad u(\vect{x},\theta)|_{x_2=1} = 1-x_1
\]
\[
\frac{\partial u(\vect{x},\theta)}{\partial x_1}\Big|_{x_1=0} = 0,\quad \frac{\partial u(\vect{x},\theta)}{\partial x_1}\Big|_{x_1=1} = 0
\]
This PDE provides a simple model of steady-state flow in porous media. The coefficient $c$ represents the permeability of a porous medium while $u$ represents the pressure head.  In this inverse problem, the objective of interest is to infer the unknown diffusion coefficient  conditioned on observation data where Bayesian approach can be naturally adopted.  A log-Gaussian process prior is given to the diffusivity field $c(\vect{x})$ with mean zero and an isotropic squared-exponential covariance kernel:
\[
C(\vect{x}_1,\vect{x}_2) = \sigma^2\exp\left(-\frac{\|\vect{x}_1-\vect{x}_2\|^2_2}{2l^2}\right)
\]
for which we choose variance $\sigma^2 =1$ and a length scale $l=0.2$. With this prior, the field can be easily parameterized with a Karhunen-Loeve (K-L) expansion:
\[
c(\vect{x},\theta) \approx \exp\left(\sum_{i=1}^d\theta_i\sqrt{\lambda_i}v_i(\vect{x})\right)
\]
where $\lambda_i$ and $v_i(\vect{x})$ are the eigenvalues and eigenfunctions of the integral operator defined by the kernel $C$, and the parameter $\theta_i$ are endowed with independent standard normal priors, $\theta_i\sim\mathcal{N}(0,1)$, which are the targets of inference. To reduce the dimension of this inference problem, the Karhunen-Loeve expansion is truncated at the first five modes ($d=5$) and the corresponding mode weights $(\theta_1,\ldots,\theta_5)$ are conditioned on data. Data are generated by combining observations of the solution field (solve the PDE on a finer 51-by-51 grid) on a uniform $11\times 11$ grid covering the unit square with additive independent Gaussian noise.
\[
y_j = u(\vect{x}_j,\theta) + \epsilon_j,\quad \epsilon_j\sim\mathcal{N}(0,0.1^2)
\]
Consistent with results from previous examples, sgHMC performs substantially better than HMC (Table \ref{tab:ePDE}). 

\begin{figure}
\begin{tabular}{cc}
\includegraphics[width=0.45\textwidth, height=0.3\textwidth]{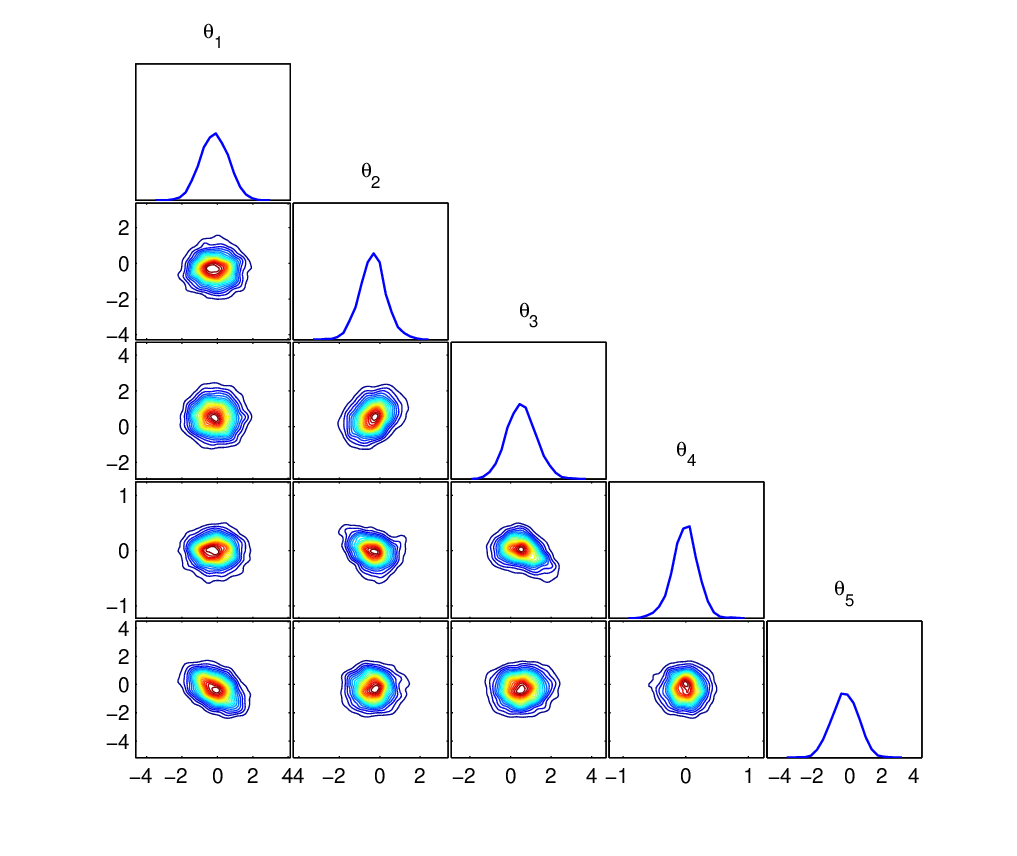}&
\includegraphics[width=0.45\textwidth, height=0.3\textwidth]{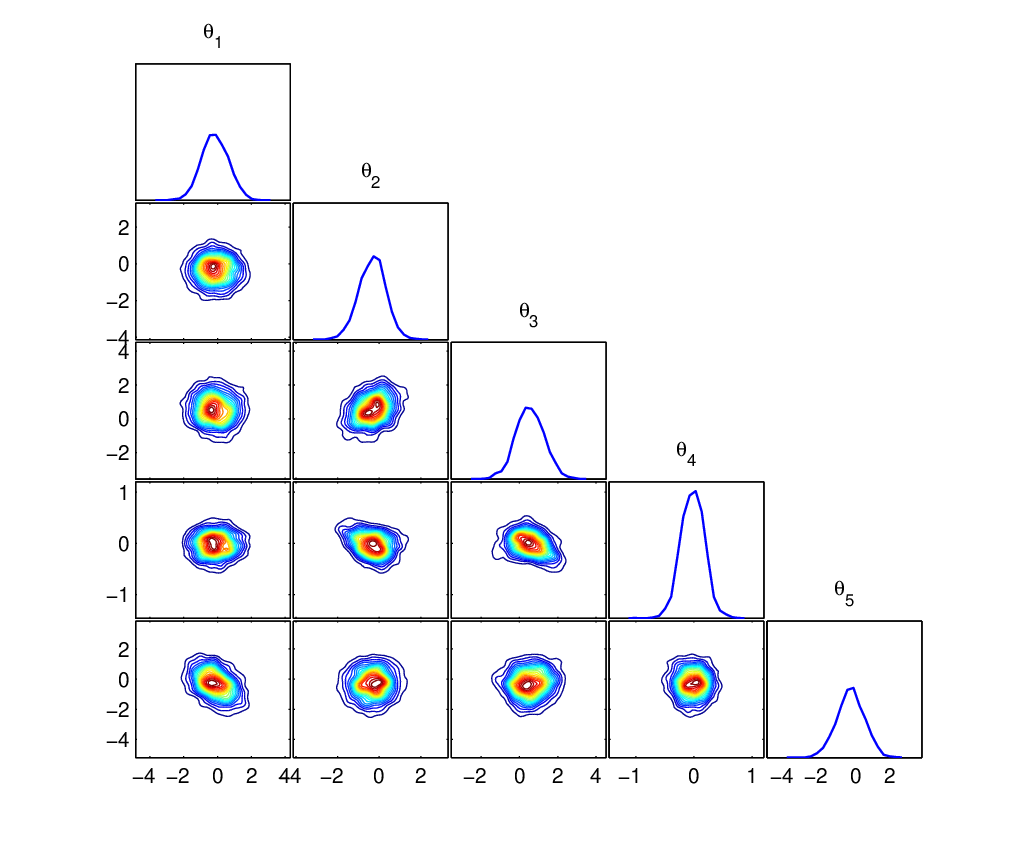} \\
HMC  &  sgHMC  
\end{tabular}
\caption{HMC vs sgHMC: an elliptic PDE inverse problem}\label{fig:ePDE}
\end{figure}

\begin{table}[!tp]
\begin{center}
\caption{Comparing HMC with sgHMC using an elliptic PDE inverse problem. For each method, we provide the acceptance rate (AR), the CPU time (s) for each iteration and the time-normalized ESS}\label{tab:ePDE}
\begin{tabular}{|c|c|c|c|c|}
\hline
Method &AR &  ESS  & s/Iteration & min ESS/s   \\\hline
HMC&$0.7719$   &  $(991.8,2091.2,2831.0)$  &  $2.02E\text{-}1$   & $1.5343$ \\
sgHMC&$0.6141$   &  $(855.7,1325.7,1937.5)$  & $6.1952E\text{-}2$  & $4.3165$\\\hline
\end{tabular}
\end{center}
\end{table}

\subsection{Computational Efficiency as Dimension Increases}
As the dimensionality of parameter space increases, the efficiency of the sparse grid interpolation decreases in general. More specifically, it requires more grid points to maintain the quality of approximation which in turn mitigates the benefit of using sparse grid interpolation. To investigate the performance of sgHMC under different dimensionality, we apply it to large scale ($N=10^5$) logistic regression models in different dimensions. We choose the step size to keep the acceptance rate around $70\%$ for HMC and collect $4000$ samples after $1000$ burn-in iterations. Both algorithms are run $10$ times and averaged to reduce the random effects on the results.

\begin{figure}[!tp]
\begin{tabular}{cc}
\includegraphics[width=0.45\textwidth]{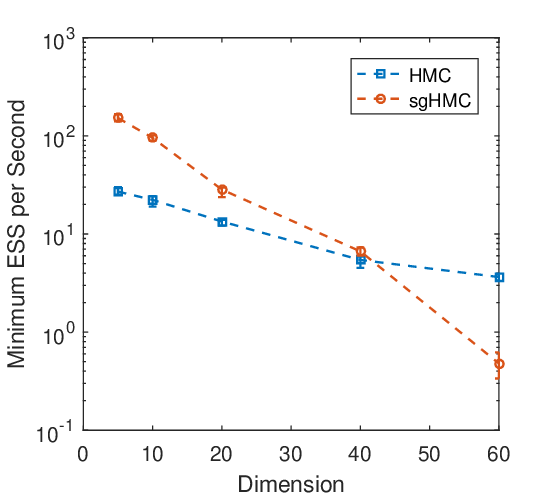}&
\includegraphics[width=0.45\textwidth]{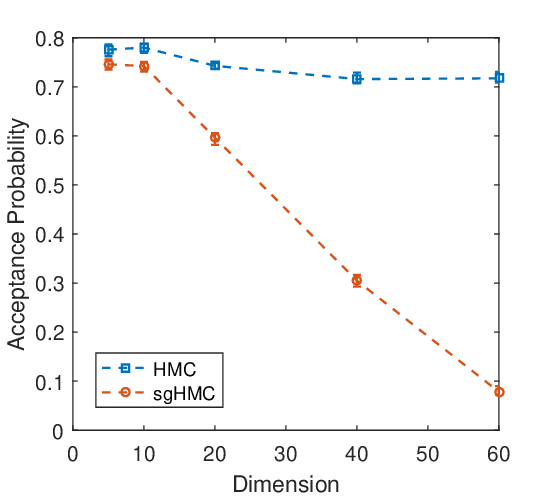} \\
\end{tabular}
\caption{Comparing HMC and sgHMC under different dimensionality on logistic regression models}\label{fig:comp}
\end{figure}

From Figure \ref{fig:comp}, we can see that sgHMC maintains efficient in mediate dimensions. As the dimensionality increases, the efficiency of sparse grid interpolation drops and the computation gain in speed eventually fails to offset the loss in approximation accuracy (see the acceptance probability in the left panel). At a dimensionality around $40$, HMC surpasses sgHMC in our current approach on these examples.

{\bf{Remark:}} Note that the efficiency of sparse grid interpolation only depends on the regularity of the target function and the dimensionality of parameter space, which makes it scalable to large scale learning problems and computationally intensive models. Even though so far our sgHMC algorithm can not generalize to extremely high dimensional problems, it can find applications on many important problems with moderate dimensionality and expensive function evaluations, such as learning hyper-parameters and Bayesian uncertainty quantification for differential equations.

\section{Approximate Target Distribution}\label{approx}
So far, our proposed method has been based on using the exact target distribution and approximating the proposal generating mechanism only. We can improve the computation speed even more by using grid approximation for $U$ in the correction step (accept/reject step) as well. In this case, the resulting sampler actually samples from an approximate distribution \[Q(q)\propto \exp(-\tilde{U}(q))\] instead of the target posterior distribution $P(q)$. The bound of the difference between these two distributions measured by the Kullback-Leibler divergence is shown in the following theorem. 

\begin{theorem}\label{th:KL} If ~$U$ and~ $V$ are energy functions corresponding to probability distributions $P$ and $Q$
\[
P(q)\propto \exp(-U(q)),\quad Q(q)\propto \exp(-V(q))
\]
then the Kullback-Leibler divergence between $P$ and $Q$ is bounded by
\[
D_{KL}(P\|Q) \leq  2\|U-V\|_{\infty}
\]
\end{theorem}
\begin{proof}
\begin{align*}
D_{KL}(P\|Q) = & \int_{\mathbb{R}^D}P(q)\ln\left(\frac{P(q)}{Q(q)}\right)\;dq\\
= & \int_{\mathbb{R}^D}P(q) (V(q)-U(q))\;dq +  \int_{\mathbb{R}^D}P(q) \ln\left(\frac{I_Q}{I_P}\right)\;dq
\end{align*}
where 
\[
I_P =\int_{\mathbb{R}^D}\exp(-U(q))\;dq,\quad I_Q =\int_{\mathbb{R}^D}\exp(-V(q))\;dq
\]
since
\begin{align*}
I_Q =& \int_{\mathbb{R}^D}\exp(-V(q))\;dq = \int_{\mathbb{R}^D}\exp(-U(q))\cdot\exp(-(V(q)-U(q)))\;dq\\
\leq \;& \exp(\|V-U\|_{\infty})\cdot\int_{\mathbb{R}^D}\exp(-U(q))\;dq = \exp( \|V-U\|_{\infty})\cdot I_P
\end{align*}
we have
\[
D_{KL}(P\|Q) \leq  \|V-U\|_{\infty}\cdot \int_{\mathbb{R}^D}P(q)\;dq +  \|V-U\|_{\infty}\cdot \int_{\mathbb{R}^D}P(q)\;dq
 = 2\|V-U\|_{\infty}
\]
\end{proof}

Note that if the potential energy function $U$ is a smooth function, $\|\tilde{U}-U\|_{\infty}\rightarrow 0$ as the grid size goes to $0$. By Theorem.~\ref{th:KL}, the resulting sampler will eventually converge to the target sampler.

We apply this method, called GHMC-complete, to the logistic regression and banana-shaped distribution examples discussed above. The results are shown in Figures \ref{fig:clog2} and \ref{fig:cbanana2}. As we can see, the posterior samples given by GHMC-complete in both cases match the exact samplers (HMC and GHMC) quite well. With appropriate grid size (around the step size), GHMC-complete can provide a high quality approximation to the standard HMC sampler. At the same time, computational efficiency has been substantially improved due to the fast computation of potential energy function in the correction step (Tables \ref{tab:clog2} and \ref{tab:cbanana2}). 

For the logistic regression example, Figure \ref{fig:AvgLogLike} shows the prediction accuracy vs. the run time for the three algorithms based on a test set. As we can see, the prediction accuracy (measured in terms of the average log-likelihood on the test data) of GHMC-complete increases faster compared to the other two methods. For computationally intensive models, the advantage of GHMC-complete will be more significant since the computation cost of the potential energy function becomes more expensive.

\section{Discussion}\label{discussion}
Due to its ability of producing distant proposals with high acceptance probability, HMC can provide rapid exploration of the parameter space when sampling from the posterior distribution. However, the gradient computation to obtain essential geometric information prevents its application on computationally intensive problems when the data size is large. To address this issue, we have proposed a relaxed framework, where HMC can take advantage of the smoothness of the potential energy function $U$ in parameter space to accelerate computation by using grid-based precomputing strategies. The key idea is to approximate the force field generated by the potential energy function $U$ through interpolation of those precomputed field at grid points in each HMC iteration. Based on these ideas, two simple grid based algorithms, Naive Grid HMC and Sparse Grid HMC, are proposed and evaluated on several problems. Empirical results show that our approach can capture the main information needed for HMC's implementation at a lower computational cost. As a result, our method tends to be more effective than standard HMC.

\begin{figure}[t]
\begin{tabular}{ccc}
\includegraphics[width=0.32\textwidth]{figs/HMC_Original_logit.eps}&
\includegraphics[width=0.32\textwidth]{figs/HMC_gridmesh_logit.eps} &
\includegraphics[width=0.32\textwidth]{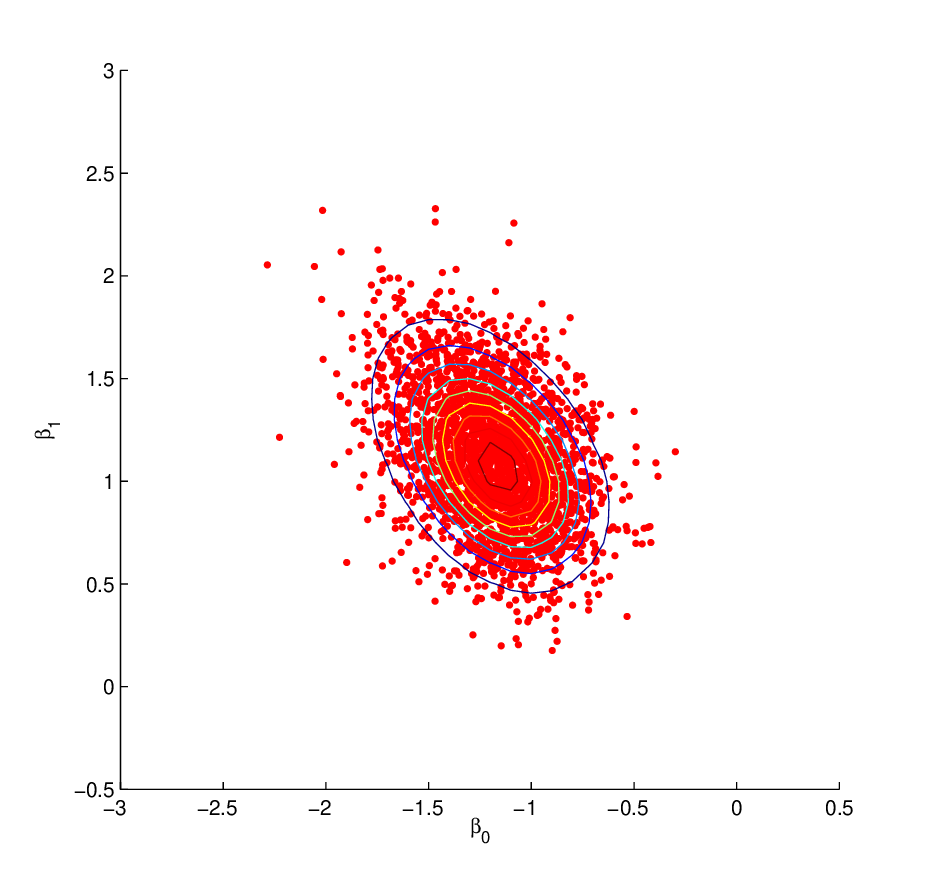} \\
HMC  &  GHMC  & GHMC-complete
\end{tabular}
\caption{HMC vs GHMC: logistic regression}\label{fig:clog2}
\begin{tabular}{ccc}
\includegraphics[width=0.32\textwidth]{figs/HMC_Original_banana.eps}&
\includegraphics[width=0.32\textwidth]{figs/HMC_gridmesh_banana.eps} &
\includegraphics[width=0.32\textwidth]{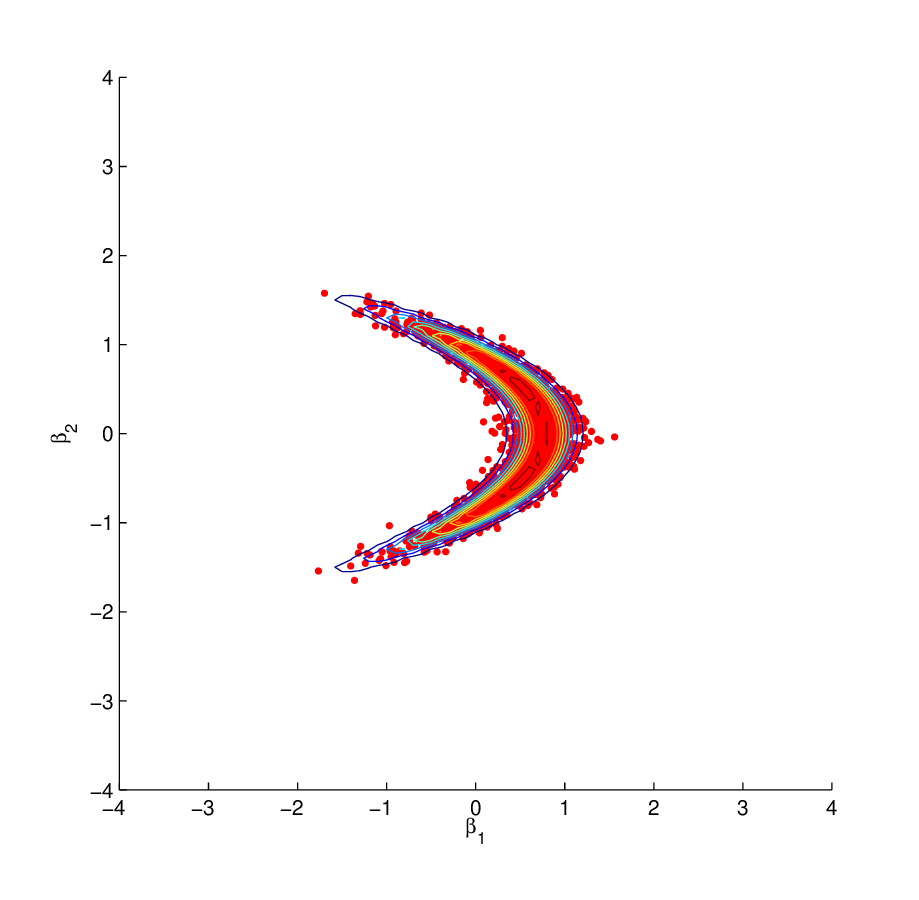} \\
HMC  &  GHMC  & GHMC-complete
\end{tabular}
\caption{HMC vs GHMC: banan-shaped distribution}\label{fig:cbanana2}
\vspace{-10pt}
\end{figure}

\begin{table}[t]
\begin{center}
\caption{Comparing HMC with GHMC using a logistic regression model. For each method, we provide the acceptance rate (AR), the CPU time (s) for each iteration and the time-normalized ESS}\label{tab:clog2}
\begin{tabular}{|c|c|c|c|c|}
\hline
Method &AR &  ESS($\beta_0,\;\beta_1$)  & s/Iteration & min ESS/s   \\\hline
HMC&$0.9225$    &  $(3200,3200)$  &  $7.0157E\text{-}4$   & $1425.3707$ \\
GHMC&$0.7981$     &  $(3200,3200)$  & $3.318E\text{-}4$  & $3013.9031$\\
GHMC-complete & $0.7931$ & $(3191.8,3200)$ & $2.9237E\text{-}4$ & 3411.5275\\\hline
\end{tabular}
\end{center}
\begin{center}
\caption{Comparing HMC with GHMC using a banana-shaped distribution model. For each method, we provide the acceptance rate (AR), the CPU time (s) for each iteration and the time-normalized ESS}\label{tab:cbanana2}
\begin{tabular}{|c|c|c|c|c|}
\hline
Method &AR &  ESS($\beta_1,\;\beta_2$)  & s/Iteration & min ESS/s   \\\hline
HMC&$0.9353$   &  $(2403,1191.6)$  &  $3.8703E\text{-}4$   & $962.1346$ \\
GHMC&$0.6587$   &  $(893.8862,766.2423)$  & $1.4498E\text{-}4$  & $1651.5917$\\
GHMC-complete & $0.6697$ & $(980.1443,796.2977)$ &  $1.2279E\text{-}4$  & $2026.6108$\\\hline
\end{tabular}
\end{center}
\end{table}

\begin{figure}[t]
\begin{center}
\includegraphics[width=0.5\textwidth]{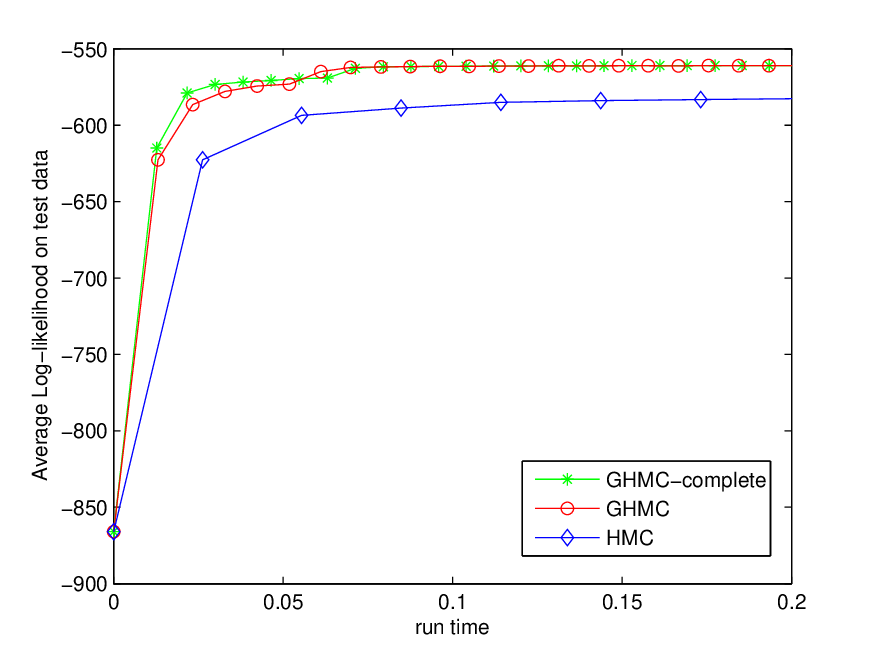}
\end{center}
\caption{The average log-likelihood on test data for a logistic regression model}\label{fig:AvgLogLike}
\end{figure}

\begin{figure}[t]
\begin{tabular}{cc}
\includegraphics[width=0.45\textwidth, height=0.3\textwidth]{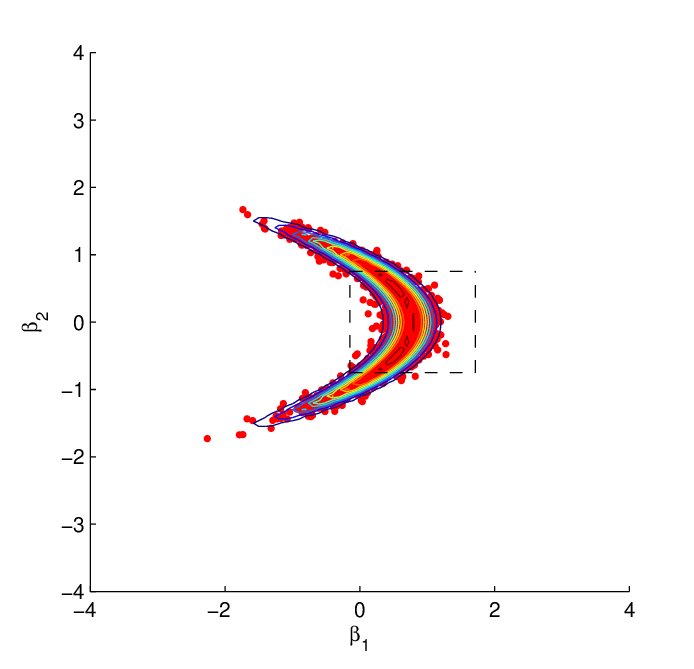} &
\includegraphics[width=0.45\textwidth, height=0.3\textwidth]{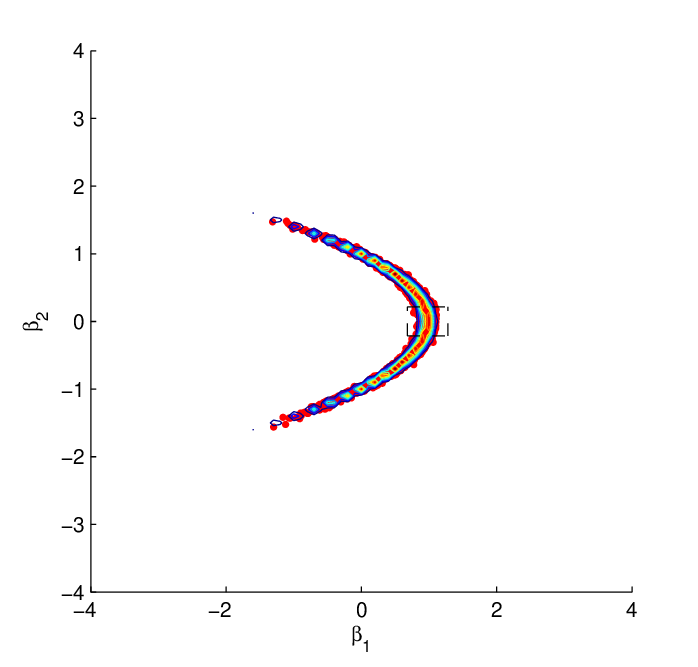} \\
 (c) \small{$N=100$} & (d) \small{$N=1000$} 
\end{tabular}
\caption{Domains of Interest via Laplace's approximation for the banana shaped distribution.}\label{laplaceFail}
\begin{tabular}{cc}
\includegraphics[width=0.45\textwidth, height=0.3\textwidth]{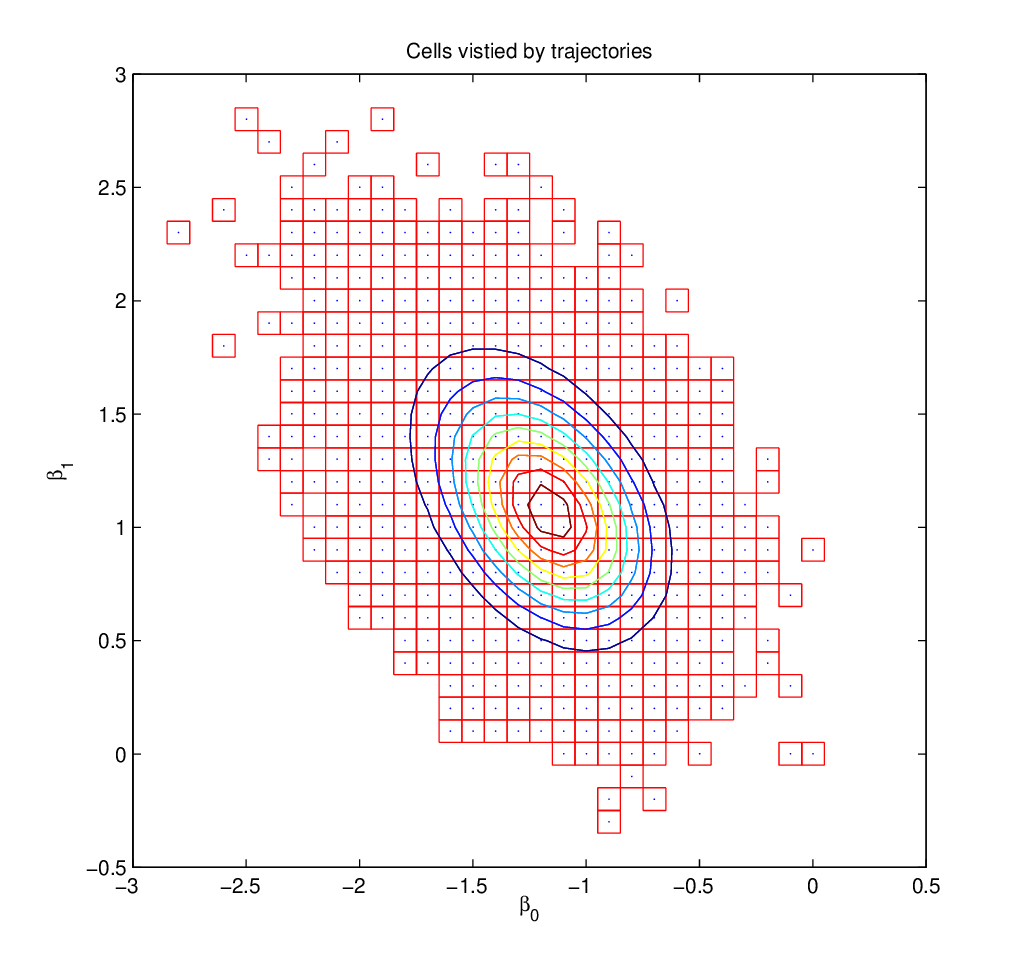}&
\includegraphics[width=0.45\textwidth, height=0.3\textwidth]{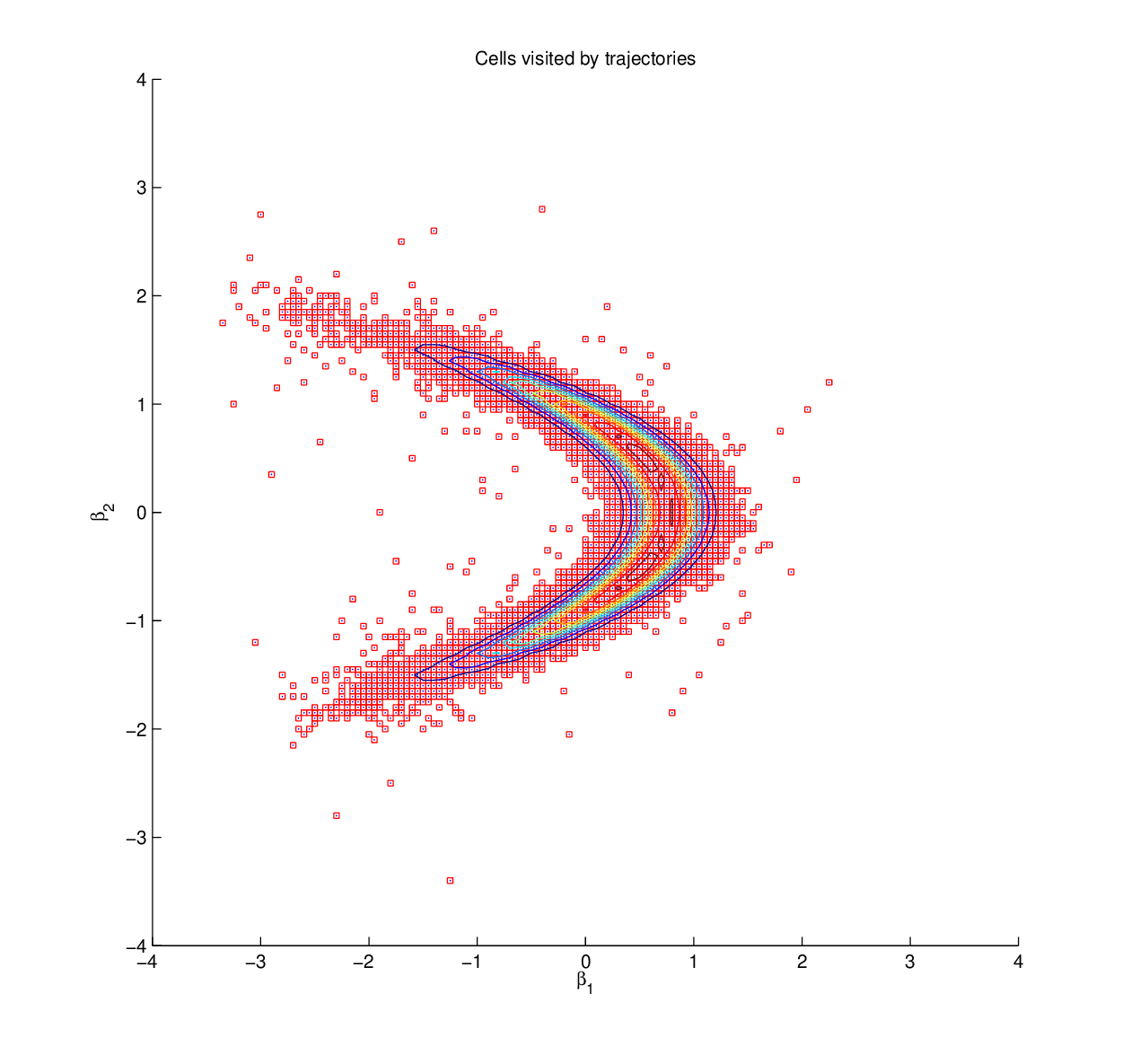} \\
(a) logistic regression  & (b) banana shaped distribution
\end{tabular}
\caption{Domains of Interest via early trajectories.}\label{fig:ROI_trajectory}
\end{figure}

While quite effective in relatively low dimensional problems, extension of grid-based HMC to high dimensional problems could be quite challenging. Future research direction could involve finding effective strategies to alleviate this issue. 

Another direction is to find more efficient method to locate the domain of interest. In subsection \ref{subsec:ROI} we used Laplace's approximation for this purpose. As shown in Figure \ref{laplaceFail}, this strategy might not be effective when the resulting Gaussian distribution is not a good approximation for the target distributions. An alternative and more general approach is based on following the trajectories of the burn-in samples. Even though MCMC samplers might not converge to the target distribution in the early stage, those trajectories can capture the high density region to some extent. Figure \ref{fig:ROI_trajectory} shows the cells visited by those early trajectories for the logistic regression example and the banana shaped distribution example.

The proposed precomputing strategy is not limited to HMC only. In fact, it can be integrated with other MCMC methods involving expensive computation of redundant information. For example, Fisher information matrices can be precomputed at each cell center to accelerate Riemannian Manifold HMC \cite{girolami11}. 

\section*{Acknowledgments}
This work is supported by NIH grant R01AI107034 and NSF grants DMS-1418422 and DMS-1622490. In addition, we appreciate the discussions with S. Lan and T. Chen. 
\clearpage
\appendix
\section{Convergence to the correct distribution}\label{sec:correct}
In order to prove that the equilibrium distribution remains the same, it suffices to show that the detailed balance condition still holds. Note that the alternative Hamiltonian $\tilde{H}(q,p)$ defines a surrogate-induced Hamiltonian flow, parameterized by the trajectory length $t$, which is a map $\tilde{\phi}_t:\;(q,p)\rightarrow(q^\ast,p^{\ast})$. Here, $(q^\ast,p^\ast)$ is the end-point of the trajectory governed by the following equations
\[
\frac{dq}{dt} = \frac{\partial \tilde{H}}{\partial p} = M^{-1}p,\quad \frac{dp}{dt} = -\frac{\partial \tilde{H}}{\partial q} = -\frac{\partial \tilde{U}}{\partial q} = -\tilde{F}
\]
Denote $\theta=(q,p),\;\theta'=(q^\ast,p^\ast)=\tilde{\phi}_t(\theta)$. In the Metropolis-Hasting step, we use the original Hamiltonian to compute the acceptance probability
\[
\alpha(\theta,\theta') = \min(1,\exp[-H(\theta')+H(\theta)])
\]
therefore, 
\begin{align*}
\alpha(\theta,\theta')\mathbb{P}(d\theta) =& \alpha(\theta,\theta')\exp[-H(\theta)]d\theta\\
\stackrel{\theta=\tilde{\phi}_t^{-1}(\theta')}{=} & \min(\exp[-H(\theta)],\exp[-H(\theta')])\left|\frac{d\theta}{d\theta'}\right|d\theta' \\
=& \alpha(\theta',\theta)\exp[-H(\theta')]d\theta'\\
= &\alpha(\theta',\theta) \mathbb{P}(d\theta')
\end{align*}
since $ \left|\frac{d\theta}{d\theta'}\right| = 1$ due to the volume conservation property of the surrogate induced Hamiltonian flow $\tilde{\phi}_t$. Now that we showed the detailed balance condition is satisfied, along with the reversibility of the surrogate induced Hamiltonian flow, the modified Markov chain will converge to the correct target distribution.

\section{More On Sparse Grid}
\subsection{Construction of sparse grid and basis functions}
The Clenshaw-Curtis type sparse grid $H^{CC}$ introduced in subsection \ref{sec:sparsegrid} is constructed from the following formulas. Here, the $x_j^i$ are defined as
\[
x_j^i =\left\{\begin{array}{ll}(j-1)/(m_i-1),\; &j = 1,\ldots,m_i, m_i > 1\\ 0.5, & j=1, m_i =1 \end{array}\right.
\]
In order to obtain nested sets of points, the number of nodes is given by 
\[
m_1 = 1\quad \text{   and   }\quad m_i = 2^{i-1} + 1 \text{   for   } i > 1
\]
Piecewise linear basis functions $a$ can be used for the univariate interpolation formulas $U^i(f)$.
\[
a_1^1(x) = 1,\quad a_j^i(x) = \left\{\begin{array}{ll} 1-(m_i-1)\cdot|x-x_j^i|,\;& |x-x_j^i| < 1/(m_i-1),\\ 0, & \text{otherwise} \end{array}\right.
\] 
for $i>1$ and $j =1,\ldots,m_i$.

\begin{figure}
\begin{tabular}{cc}
\includegraphics[width = .45\textwidth, height = .45\textwidth]{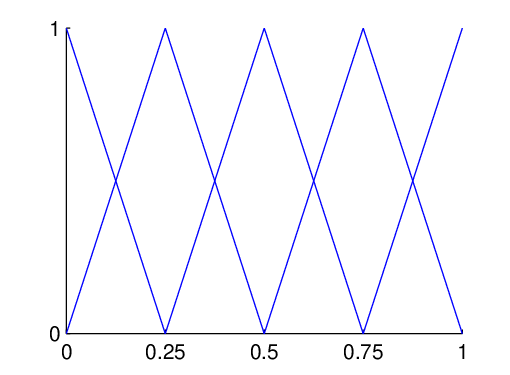} &
\includegraphics[width = .45\textwidth, height = .45\textwidth]{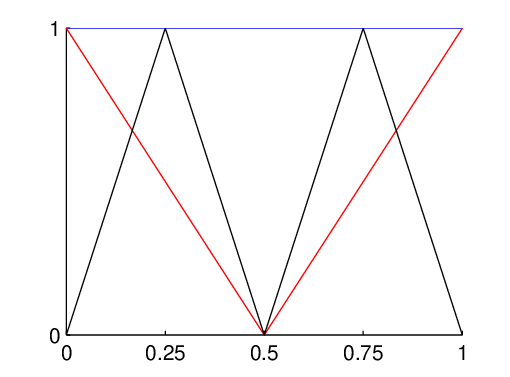}\\
(a) Nodal basis  &  (b) Hierarchical basis
\end{tabular}
\caption{Piecewise linear Nodal basis (a) and hierarchical functions (b) with support nodes $x_j^i\in X_{\Delta}^i,i=1,2,3$ for the Clenshaw-Curtis grid}\label{fig:nodalhierar}
\end{figure}

\subsection{Derivation of the hierarchical form}
With the selection of nested sets of points, we can easily transform the univariate nodal basis into the hierarchical one. By definition, we have
\begin{align*}
\Delta^i(f) &= U^i(f) - U^{i-1}(f) \\
                  &= \sum_{j=1}^{m_i}f(x_j^i)\cdot a_j^i - \sum_{j=1}^{m_i} U^{i-1}(f)(x_j^i)\cdot a_j^i\\
                  &= \sum_{j=1}^{m_i}\big(f(x_j^i) - U^{i-1}(f)(x_j^i)\big) \cdot a_j^i
\end{align*}
since $f(x_j^i) - U^{i-1}(f)(x_j^i) = 0,\; \forall \; x_j^i \in X^{i-1}$,
\begin{equation}\label{eq:unihier}
\Delta^i(f) = \sum_{x_j^i\in X^i_{\Delta}}\big(f(x_j^i)-U^{i-1}(f)(x_j^i)\big)\cdot a_j^i
\end{equation}
From \eqref{eq:unihier} we note that for all $\Delta^i(f)$, only the basis functions belonging to the grid points that have not yet occurred in a previous set $X^{i-k},\;1\leq k \leq i-1$ are involved. Fig.\ref{fig:nodalhierar} gives a comparison of the nodal and the hierarchical basis functions and Fig.\ref{fig:interp} shows the construction of the interpolation formula using nodal basis functions and function values versus using hierarchical basis functions and hierarchical surpluses for a univariate function $f$ . Both figures are reproductions based on Klimke and Wohlmuth \cite{klimke05}.

\begin{figure}
\begin{tabular}{cc}
\includegraphics[width = .45\textwidth, height = .45\textwidth]{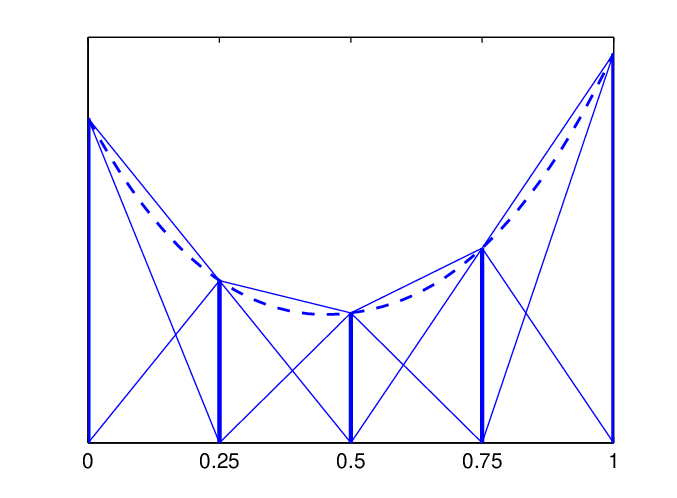} &
\includegraphics[width = .45\textwidth, height = .45\textwidth]{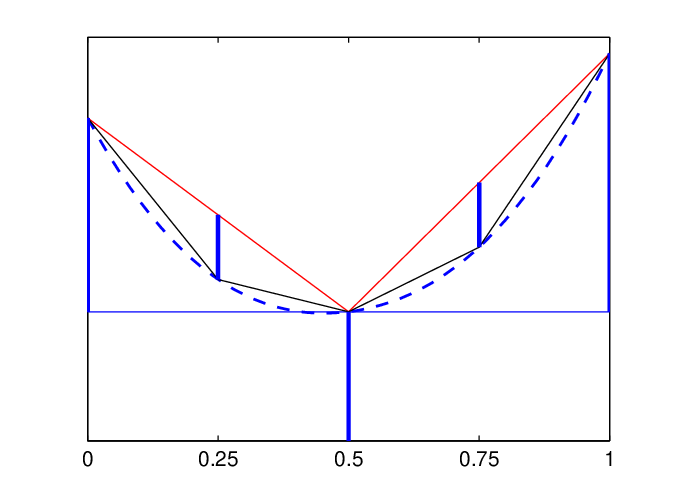}\\
(a) Nodal basis  &  (b) Hierarchical basis
\end{tabular}
\caption{Piecewise linear interpolation: Nodal versus Hierarchical}\label{fig:interp}
\end{figure}

Applying the tensor product formula \eqref{eq:tensor} with $\Delta^i$ given in \eqref{eq:unihier}, the hierarchical update in the Smolyak algorithm \eqref{eq:smolyak} now can be rewritten as
\begin{align*}
\Delta A_{q,d}(f) &= \sum_{|\vect{i}|=q}(\Delta^{i_1}\otimes\cdots\otimes\Delta^{i_d})(f)\\
                             &= \sum_{|\vect{i}|=q}\sum_{x_{j_1}^{i_1}\in X_{\Delta}^{i_1}}\cdots\sum_{x_{j_d}^{i_d}\in X_{\Delta}^{i_d}}\big(f(x_{j_1}^{i_1},\ldots,x_{j_d}^{i_d})-A_{q-1,d}(f)(x_{j_1}^{i_1},\ldots,x_{j_d})\big)\cdot(a_{j_1}^{i_1}\otimes\cdots\otimes a_{j_d}^{i_d})
\end{align*}

\clearpage 
\bibliographystyle{siam}

\end{document}